\documentclass[english, 11pt]{article}
\usepackage{epsfig}
\usepackage[titletoc,toc,title]{appendix}
\newenvironment{proof}{\noindent {\bf Proof.}}{\bigskip}

\usepackage{epsfig}
\usepackage{amsmath}
\usepackage{amssymb}
\usepackage{color}
\usepackage{amsmath}

{ %
  \newtheorem{lemma}{\textbf{Lemma}}[section]%
  \newtheorem{theorem}[lemma]{\textbf{Theorem}}%
  \newtheorem{proposition}{Proposition}%
  {%
    }%
  {%
    }%
  {%
    }%
}

\begin{document}

\centerline{\Large\bf Algorithms for distance problems in planar}

\medskip
\centerline{\Large\bf  complexes of global nonpositive curvature}

\vspace{10mm}

\centerline{\large {\sc Daniela Maftuleac}}

\vspace{5mm}

\begin{center}
Laboratoire d'Informatique Fondamentale de Marseille,\\[0.1cm]
Universit\'e d'Aix-Marseille, \\[0.1cm]
F-13288 Marseille Cedex 9, France, \\[0.1cm]
{daniela.maftuleac@lif.univ-mrs.fr}
\end{center}


\noindent

\vspace{7mm}
\begin{footnotesize} \noindent {\bf Abstract.} CAT(0) metric spaces and hyperbolic spaces play an important role in combinatorial and geometric group theory. In this paper, we present efficient algorithms for distance problems in CAT(0) planar complexes. First of all, we present an algorithm for answering single-point distance queries in a CAT(0) planar complex. Namely, we show that for a CAT(0) planar complex $\mathcal K$ with $n$ vertices, one can construct in $O(n^2\log n)$ time a data structure $\mathcal D$ of size $O(n^2)$ so that, given a point $x\in \mathcal K,$ the shortest path $\gamma(x,y)$ between $x$ and the query point $y$ can be computed in linear time. Our second algorithm computes the convex hull of a finite set of points in a CAT(0) planar complex. This algorithm is based on Toussaint's algorithm for computing the convex hull of a finite set of points in a simple polygon and it constructs the convex hull of a set of $k$ points in $O(n^2\log n + nk \log k)$ time, using a data structure of size $O(n^2 + k)$.
\end{footnotesize}


\section{Introduction}

\subsection{CAT(0) metric spaces}

Introduced by Gromov in 1987, CAT(0) metric spaces (or spaces of global nonpositive curvature) constitute a far-reaching common generalization of Euclidean and hyperbolic spaces and simple polygons.
CAT initials stand for Cartan-Alexandrov-Toponogov, who made substantial contributions to the theory of comparison geometry. CAT(0) spaces are precisely the complete Hadamard spaces \cite{Ba, BrHa}.


The impact of CAT(0) geometry on mathematics plays a significant role, especially in the field of geometric group theory \cite{BrHa}. The particular case of CAT(0) cube complexes has received lately much attention from the geometric group theory community \cite{BrHa, Gr, HaWi}.
It is clear that CAT(0) geometry has numerous applications. A notable example of this appears in the paper of Billera, Holmes and Vogtmann \cite{BiHoVo} which shows that the space of all phylogenetic trees with the same set of leaves can be seen as a CAT(0) cube complex. In their paper, CAT(0) geometry
is used to solve problems of qualitative classification in biological systems.\\
Other important examples of CAT(0) cube complexes come from reconfigurable systems \cite{GhPe}, a large family of systems which change according to some local rules, e.g. robotic motion planning, the motion of non-colliding particles in a graph, and phylogenetic tree mutation, etc.
In many reconfigurable systems, the parameter space of all possible positions of the system can be seen as a CAT(0) cube complex \cite{GhPe} and computing geodesics in these complexes is equivalent to finding the optimal way to get the system from one position to another one under the corresponding metric.

\subsection{Algorithmic problems on CAT(0) metric spaces}

Presently, most of the known results on CAT(0) metric spaces are mathematical. To the best of our knowledge, from the algorithmic point of view, these spaces remain relatively unexplored. Still there are some algorithmic results in some particular CAT(0) spaces.
In her doctoral thesis \cite{Ow}, Owen proposed exponential time algorithms for computing the shortest path in the space of phylogenetic trees.\\
Subsequently, the question of whether the distance and the shortest path between two trees in this CAT(0) space can be computed in polynomial (with respect to the number of leaves) time was raised. Recently, Owen and Provan \cite{OwPro} solved this question in the affirmative; the paper \cite{ChaHo} reports on the implementation of the algorithm of \cite{OwPro}.
Using the result of \cite{OwPro}, Ardila, Owen, and Sullivant \cite{ArOwSu} described a finite algorithm that computes the shortest path between two points in general CAT(0) cubical complexes. This algorithm is not a priori polynomial and finding such an algorithm that computes the shortest path in a CAT(0) complex of general dimension, remains an open question. In the paper \cite{ChMa}, we proposed a polynomial time algorithm for two-points shortest path queries in  2-dimensional CAT(0) cubical complex and some of its subclasses.
In this paper, we present efficient algorithms for single-point distance queries and convex hulls in CAT(0) planar complexes. A detailed description of these results is given in the author's doctoral thesis \cite{Ma}.

CAT(0) planar complexes have numerous structural and algorithmic properties from their underlying graphs \cite{Ch_CAT}.
In \cite{BaPe2}, Baues and Peyerimhoff give a combinatorial characterization of the non-positive curvature of the tilings of planar graphs.
Later, the same authors \cite{BaPe}, characterized the geodesics in a tiling of non-positive curvature. They use these characterizations to estimate the growth of distance balls, Gromov hyperbolicity and 4-colorability of certain classes of tilings of non-positive curvature.
Chepoi, Dragan and Vax\`es have studied algorithmic problems for routing as well as calculating of the center and the diameter in a CAT(0) complex and the underlying graphs \cite{ChDrVa_soda, ChDrVa_jalg}.

\subsection{Shortest path problem}

The shortest path problem is one of the best-known algorithmic
problems with many applications in routing, robotics, operations
research, motion planning, urban transportation, and terrain navigation.
This fundamental problem has been intensively studied both in discrete
settings like graphs and networks (see, e.g., Ahuja, Magnanti, and
Orlin \cite{AhMaOr}) as well as in geometric spaces (simple polygons,
polygonal domains with obstacles, polyhedral surfaces, terrains; see,
e.g., Mitchell \cite{Mi}).
Several algorithms for computing shortest paths inside a simple polygon are
known in the literature \cite{GuiHer, GuiHerLeShaTar, HerSu, LeePre,ReSto},
and all are based on a triangulation of  $P$ in a preprocessing step (which
can be done in linear time due to Chazelle's algorithm \cite{Cha}).
The algorithm of Lee and Preparata \cite{LeePre}  finds the shortest path
between two points of a triangulated simple polygon in linear time
({\it two-point shortest path queries}). Given a source point, the algorithm
of Reif and Storer \cite{ReSto} produces in $O(n\log n)$ time a search
structure (in the form of a shortest path tree) so that the shortest path
from any query point to the source can be found in time linear in the number
of edges of this path (the so-called {\it single-source shortest path queries}).
Guibas et al. \cite{GuiHerLeShaTar}  return a similar search structure,
however their preprocessing step takes only linear time once the polygon is
triangulated (see Hersberger and Snoeyink \cite{HerSno} for a significant
simplification of the original algorithm of \cite{GuiHerLeShaTar}). Finally,
Guibas and Hershberger \cite{GuiHer} showed how to preprocess a triangulated
simple polygon $P$ in linear time to support shortest-path queries between any
two points $p,q\in P$ in time proportional to the number of edges of the
shortest path between $p$ and $q.$ Note that the last three aforementioned algorithms
also return in $O(\log n)$ time the distance between the queried points.
In the case of shortest path queries in general polygonal domains $D$  with holes,
the simplest approach is to compute at the preprocessing step the visibility
graph of $D.$ Now, given two query points $p,q,$ to find a shortest path between
$p$ and $q$ in $D$ (this path is no longer unique), it suffices to compute this path
in the visibility graph of $D$ augmented with two vertices $p$ and $q$ and all
edges corresponding to vertices of $D$ visible from $p$ or $q;$ for a detailed
description of how to efficiently construct the visibility graph, see the survey
\cite{Mi} and the book \cite{BeChKrOv}.  An alternative paradigm is the so-called
{\it continuous Dijkstra} method, which was first applied  to the shortest path
problem in general polygonal domains by Mitchell \cite{Mi1} and subsequently
improved to a nearly optimal algorithm by Hershberger and Suri \cite{HerSu};
for an extensive overview of this method and related references, see again
the survey by Mitchell \cite{Mi}.\\

\subsection{Convex hull algorithms}

Another algorithmic problem presented in this paper is computing the convex hull of a finite set of points.
In the Euclidean plane, many algorithms solve optimally this problem \cite{Gra,Ja,Ka,PrHo,PrSh}.
The most common method for computing the convex hull of a finite point-set is the incremental method \cite{PrSh}.
In 3-dimensional Euclidean space, there exist efficient algorithms that construct the convex hull of a point-set \cite{PrSh}.
For simple polygonal domains equipped with the intrinsic $l_2$-metric, Toussaint describes an algorithm for constructing the convex hull of a finite set of points \cite{To}.

The recent paper by Fletcher et al. \cite{FlMoPhVen} investigates algorithmic questions related to computing approximate convex hulls and centerpoints of point-sets in the CAT(0) metric space $P(n)$ of all positive definite $n\times n$ matrices.\\
In their recent paper \cite{ArNi}, Arnaudon and Nielsen, describe a generalization of a 1-center algorithm \cite{BaCl} to arbitrary Riemannian geometries, especially the space of symmetric positive definite matrices.\\

\subsection{Structure of the paper}

In this paper, we present efficient algorithms for  single-point distance queries and convex hulls in CAT(0) planar complexes. A detailed description of these results is given in the author's doctoral thesis \cite{Ma}.
First, we give an efficient algorithm for answering one-point distance queries in CAT(0) planar complexes. Namely, we show that for a CAT(0) planar complex $\mathcal K$ with $n$ vertices, one can construct
in $O(n^2\log n)$ time a data structure $D$ of size $O(n^2)$ so that, given a point $x\in\mathcal K$, the shortest path $\gamma(x, y)$ between $x$ and the query point $y$ can be computed in linear time. Second we propose an algorithm for computing the convex hull of a finite set of points in a CAT(0) planar complex in $O(n^2\log n+nk\log k)$ time, using a data structure of size $O(n^2+k)$. This algorithm is based on Toussaint's algorithm for computing the convex hull of a finite set of points in a simple polygon.

The remaining part of the paper is organized as follows. In the preliminary section, we introduce CAT(0) metric spaces and CAT(0) planar complexes. We also formulate the single-point shortest path query problem and the convex hull problem.
In Section 3, we define the shortest path map SPM($x$) of a given source-point $x$ in a CAT(0) planar complex as a partition of the complex into convex sets (cones), such that the shortest paths connecting $x$ and points of the same set, are equivalent from the combinatorial point of view.
We present an efficient algorithm which for a give source-point $x$ computes the shortest path map SPM($x$) in a CAT(0) planar complex.
In Section 5, we present the detailed description of the algorithm for constructing the shortest path map in a CAT(0) planar complex and of the data structure used by this algorithm.
Given a CAT(0) planar complex and a source-point $x$, we use the shortest path map of origin $x$ so that for any query point $y$ it is possible to determine the cone containing $y$. We show how to compute the unfolding of a cone in ${\mathbb R}^2$ efficiently, and we construct the shortest path between $x$ and $y$ as the pre-image of the Euclidean geodesic between the images of $x$ and $y$ in ${\mathbb R}^2$.
In the Section 6, we present an efficient algorithm for computing the convex hull of a finite point-set in CAT(0) planar complexes, based on Toussaint's algorithm for convex hulls in a simple polygon. This algorithm computes the convex hulls of all subsets belonging to different regions of SPM($x$). Further, it constructs a weakly-simple polygon $P$ which is the convex hulls of subsets connected by a geodesic segments. Finally, the boundary of the convex hull of $P$ is exactly the shortest path between a point of the convex hulls and its copy in $\mathcal K\setminus P.$

\section{Preliminaries}

\subsection{CAT(0) metric spaces}\label{prel1}
Let $(X,d)$ be a metric space. A \emph{geodesic}
joining two points $x$ and $y$ from $X$ is the image of a
(continuous) map $\gamma$ from a line segment $[0,l]\subset \mathbb{R}$
to $X$ such that $\gamma(0)=x, \gamma(l)=y$ and
$d(\gamma(t),\gamma(t'))=|t-t'|$ for all $t,t'\in [0,l].$ The space
$(X,d)$ is said to be \emph{geodesic} if every pair of points $x,y\in
X$  is joined by a geodesic \cite{BrHa}.  A \emph{geodesic triangle} $\Delta (x_1,x_2,x_3)$ in a geodesic
metric space $(X,d)$ consists of three distinct points in $X$ (the
vertices of $\Delta$) and a geodesic  between each pair of vertices
(the sides of $\Delta$). A \emph{comparison triangle} for $\Delta
(x_1,x_2,x_3)$ is a triangle $\Delta (x'_1,x'_2,x'_3)$ in the
Euclidean plane ${\mathbb E}^2$ such that $d_{{\mathbb
E}^2}(x'_i,x'_j)=d(x_i,x_j)$ for $i,j\in \{ 1,2,3\}.$ A geodesic
metric space $(X,d)$ is defined to be a \emph{$CAT(0)$ space} \cite{Gr}
if all geodesic triangles $\Delta (x_1,x_2,x_3)$ of $X$ satisfy the
comparison axiom of Cartan--Alexandrov--Toponogov:

\medskip\noindent
\emph{If $y$ is a point on the side of $\Delta(x_1,x_2,x_3)$ with
vertices $x_1$ and $x_2$ and $y'$ is the unique point on the line
segment $[x'_1,x'_2]$ of the comparison triangle
$\Delta(x'_1,x'_2,x'_3)$ such that $d_{{\mathbb E}^2}(x'_i,y')=
d(x_i,y)$ for $i=1,2,$ then $d(x_3,y)\le d_{{\mathbb
E}^2}(x'_3,y').$}

\medskip\noindent
This simple axiom turns out to be very powerful, because CAT(0) spaces can be characterized  in several natural ways (for a full account of this theory consult the book \cite{BrHa}). In particular, a geodesic metric space $(X,d)$ is CAT(0) if and only if any two points of this space can be joined by a unique geodesic. CAT(0) is also equivalent to convexity of the function $f:[0,1]\rightarrow X$ given by $f(t)=d(\alpha (t),\beta (t)),$ for any geodesics $\alpha$ and $\beta$ (which is further equivalent to convexity of the neighborhoods of convex sets). This implies that CAT(0) spaces are contractible.\\

The notion of an angle between two geodesics from a common point in a CAT(0) space is given using the Alexandrov definition of angle in an arbitrary metric space \cite{AlZa}.
The angle between the sides of a geodesic triangle $\Delta$ with distinct vertices in a CAT(0) metric space $X$ is no greater than the angle between the corresponding sides of the comparison triangle $\Delta'$ in $\mathbb{E}^2$ \cite{AlZa, BrHa}.

A subset $Y$ of a CAT(0) space $(X, d)$ is called \emph{convex} if the geodesic segment joining any two points of
$Y$ is entirely contained in $Y$.
For any finite set of points $S\subset \mathcal K$ the \emph{convex hull} conv($S$) for $S$ is the convex set, minimal by inclusion, such that $S \subset$conv$(S).$

\subsection{CAT(0) planar complexes}

A piecewice-Euclidean cell complex $\mathcal X$ consists of a collections of convex Euclidean polyhedra glued together by isometries along their faces.
$\mathcal X$ endowed with the length metric $d$ induced by the Euclidean metric on each cell is a geodesic space \cite{BrHa}.\\

A \emph{planar complex} is a 2-dimensional piecewise Euclidean cell complex whose 1-skeleton has a planar drawing in such a way that the 2-cells of the complex are exactly the inner faces of the 1-skeleton in this drawing. The \emph{1-skeleton} (or the graph) of a complex $X$ has the 0-faces  as vertices and the 1-faces of $X$ as edges. \\
\noindent A planar complex $\mathcal K$ can be endowed with the intrinsic $l_2$-metric in the following way. Suppose that inside every 2-cell of $\mathcal K$ the distance is measured according to an $l_2-$metric. A {\it path} between $x,y\in {\mathcal K}$ is a sequence of points $p=(x_0=x,x_1,\ldots,x_{k-1},x_k=y)$ such that any two consecutive points $x_i,x_{i+1}$ belong to a common cell. The length of a path is the sum of distances between all pairs of consecutive points of this path.  Then the intrinsic $l_2-$metric between two points of $\mathcal K$ equals to the infimum of the lengths of the finite paths joining them. A planar complex $\mathcal K$ is a {\it planar CAT(0) complex} if with respect to its intrinsic metric, $\mathcal K$ is a CAT(0) space. Equivalently, a planar complex is CAT(0) planar (see Fig. \ref{fig2}) if for all inner 0-cells the sum of angles in one of these points is at least equal to $ 2\pi$. A vertex of $\mathcal K$ for which this sum is greater then $2\pi$ is called \emph{vertex of negative curvature}.
By triangulating each 2-cell of a CAT(0) planar complex, we can assume without loss of generality and without changing the size of the input data, that all 2-cells (faces) of $\mathcal K$ are isometric to arbitrary triangles of the Euclidean plane.\\
The 0-dimensional cells of the complex are called \emph{vertices}, forming
the vertex set $V(\mathcal K)$ of $\mathcal K$ and the 1-dimensional cells of $\mathcal K$ are called \emph{edges} of the complex, and denoted
by $E(\mathcal K)$.
All complexes occurring in our paper are finite, i.e., they have only finitely many cells.\\
A point (respectively a vertex) $x$ of $\mathcal K$ is called \emph{inner} point (vertex) of $\mathcal K$ if $x$ does not belong to the boundary of the complex which is denoted by $\partial\mathcal K$. The same way we define inner points of an edge $(uv)$ of $\mathcal K$ as the points distinct from the endpoints of the edge $u$ or $v$.\\
For any vertex $v$ of $\mathcal K$, we will denote by $\theta(v)$ the sum of angles with origin in $v$ from the faces incident to $v$ in $\mathcal K$.
As stated above, for each inner vertex $v$ of $\mathcal K$, $\theta(v)\geq 2\pi$.\\

\begin{figure}[h]\centering
 \includegraphics[width=0.55\textwidth]{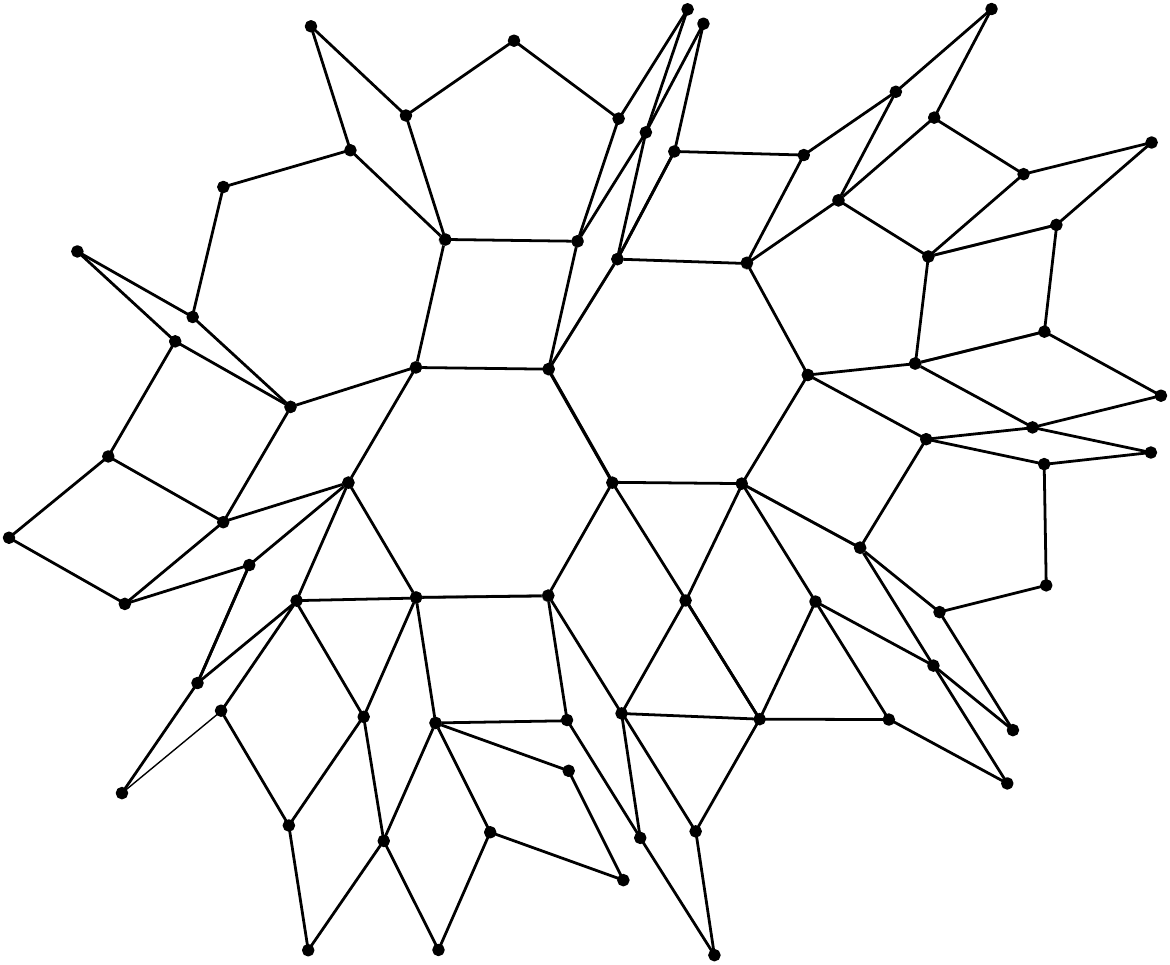}
\caption{A CAT(0) planar complex.} \label{fig2}
 \end{figure}

Let $u$ and $v$ be two points of $\mathcal K$, then the set of points $r(u,v)$ called \emph{geodesic ray of origin $u$ and of direction $v$}, is a geodesic between its origin $u$ and a point on the boundary of $\mathcal K$, so that $r(u,v)$ contains $v$.
We call \emph{line} in $\mathcal K$ any geodesic $l = \gamma(p,q)$ between two points $p$ and $q$ of the boundary of $\mathcal K.$\\
Subsequently, we consider two geodesics $\gamma(u, v)$ and $\gamma(u, w)$ in $\mathcal K$, sharing a common origin $u$ and so that the points $v, w$ belong to the boundary of the complex. Let $\alpha$ be one of two complementary angles formed by $\gamma(u, v)$ and $\gamma(u, w)$ at $u$. 
A \emph{cone} in a CAT(0) planar complex $\mathcal K$ is the set of all points of the complex located in the region bounded by the two geodesics $\gamma(u, v) and \gamma(u, w)$ containing the angle $\alpha$, where $0<\alpha<\pi$.
In other words, due to the planarity of $\mathcal K$, two geodesics forming an angle $\alpha<\pi$ determine exactly one cone containing $\alpha$ in $\mathcal K$. Note that if $u$ is an inner point of $\mathcal K$ then there exist at least three cones with the common origin $u$. We denote by $\mathcal C(u; v, w)$ the cone of origin (or apex) $u$, where $\gamma(u, v)$ and $\gamma(u, w)$ are the \emph{sides} of the cone. By the \emph{interior} of the cone $\mathcal C (u; v, w)$ we mean the set of points int$(\mathcal C (u; v, w)): = \mathcal C (u; v , w) \setminus (\gamma(u, v) \cup \gamma(u, w)).$\\

\section{Shortest Path Map}

The notion of shortest path map was introduced by Hershberger and Suri \cite{HerSu} as a preprocessing step for the \emph{continuous Dijkstra algorithm}. This algorithm computes the shortest path between a given point and any other point in a polygon with holes or in the plane in the presence of obstacles.
This method is a conceptual algorithm to compute shortest paths from a given source $s$ to all other points, by simulating the propagation of a sweeping line from a point to all points using scanning level-lines.
The output of the continuous Dijkstra method is called the \emph{shortest path map} SPM($s$) of a given point $s.$ The shortest path map SPM($s$) is a subdivision of the polygon (free-space) into cells, such that each cell is a set of points whose shortest paths to $s$ are equivalent from a combinatorial point of view.

\begin{figure}[h]\centering
 \includegraphics[width=0.55\textwidth]{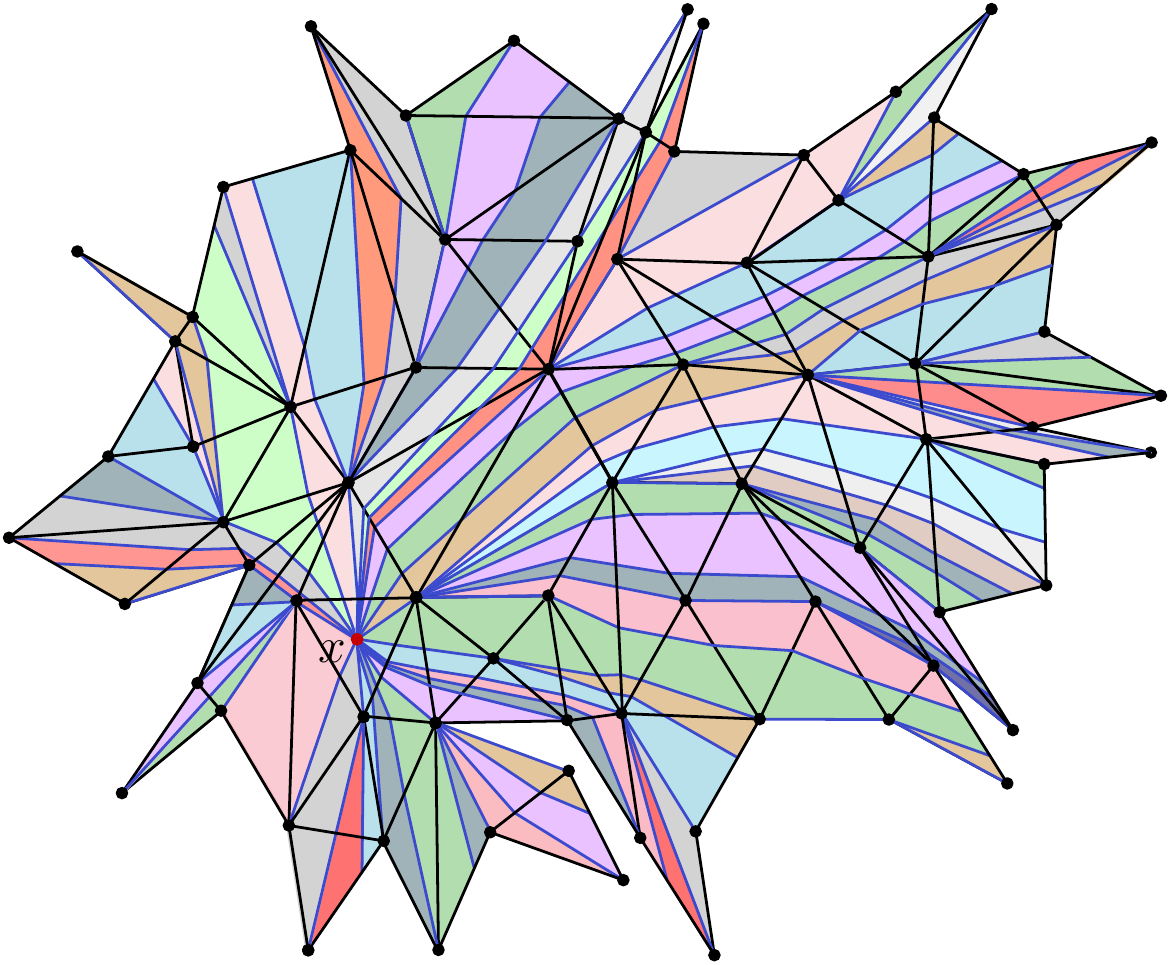}
\caption{A shortest path map in a CAT(0) planar complex (the coloring indicates the partition into convex cones).} \label{fig22}
 \end{figure}
The shortest path map SPM($x$) in a CAT(0) planar complex $\mathcal K$ is a partition of the complex in convex cones. This partition has a shortest path tree structure with the common root $x \in \mathcal K.$ The construction of SPM($x$) is used in our algorithms for computing the shortest path and the convex hull of a finite set of points in $\mathcal K$. Thus we will give a more detailed definition of SPM($x$) in $\mathcal K.$

Let $x$ be a point in $\mathcal K$.
A \emph{geodesic tree}  is the union of a finite set of geodesics $\gamma(x,y)$ with the common source $x$ and  each having the second endpoint $y$ at the
boundary of $\mathcal K$ such that the intersection of any two geodesics is a geodesic between $x$ and a vertex of $\mathcal K$.\\

The \emph{shortest path map} of origin $x$, denoted SPM($x$), is defined as the set of all geodesics connecting $x$ and boundary points of $\mathcal K$ passing through at least one vertex of the complex.
Moreover, for every vertex $z$ of negative curvature, there exists at least two geodesics $\gamma(x,p)$ and $\gamma(x,q)$ in SPM($x$) passing through $z$ and containing $\gamma(x,z)$. Each of these two geodesics form with $\gamma(x,z)$ an angle equal to $\pi$.
More formally, we can define SPM($x$) as follows.


Given a point $x \in \mathcal K$, the \emph{shortest path map SPM($x$)} is a partition of $\mathcal K$ into cones $\mathcal C(z; p_i, q_i)$, $i\in \mathbb{N}$ (see Fig. \ref{fig22}), such that \\
\hspace*{0.2cm} (1) for every vertex $y$ of $\mathcal K$ there exists a geodesic $\gamma(x, p),$ with $p$ on the boundary of $\mathcal K$ which passes via $y$ and: \\
\hspace*{0.7cm} (a) if $ \theta(y) = 2\pi$, then $y$ belongs to the common side $\gamma(z, p)$ of two cones $\mathcal C(z; p, q)$, $\mathcal C(z';p,r)$ of SPM($x$);\\ 
\hspace*{0.7cm} (b) if $y$ is a vertex of negative curvature, then $y$ is the apex of at least one cone $\mathcal C(y;p,q)$ of SPM($x$); \\
\hspace*{0.2cm} (2) let $\mathcal C(z;p,q)$ be a cone of SPM($x$), then its apex $z$ belongs to the geodesics $\gamma(x,p)$ and $\gamma(x,q).$\\

\subsection{Structure and properties of SPM($x$)}

\noindent We continue with a list of simple but essential properties of the shortest path map SPM($x$) in a CAT(0) planar complex $\mathcal K$.

\begin{proposition} \label{pr-spm}
Let $\mathcal K$ be a CAT(0) planar complex, $x$ a point of $\mathcal K$ and SPM($x$) the shortest path map of $\mathcal K$, then the following conditions are satisfied: \\
\noindent \hspace*{0.3cm} (i) The shortest path map SPM($x$) has a geodesic tree structure; \\
\noindent \hspace*{0.3cm} (ii) The angle $\angle_z(p,q)$ formed by the sides $\gamma(z, p)$ and $\gamma(z, q)$ of a cone $\mathcal C (z; p, q)$ of SPM($x$) is less 
than $\pi$;\\ 
\noindent \hspace*{0.3cm} (iii) Each cone $\mathcal C(z; p, q)$ of SPM($x$) is a convex subset of $\mathcal K$; \\
\noindent \hspace*{0.3cm} (iv) If $\mathcal C(z; p, q)$ is a cone of SPM($x$), then int$(\mathcal C(z; p, q))$ contains no vertices of $\mathcal K$; \\
\hspace*{0.3cm} (v) An inner point of $\mathcal K$ belongs either to a single cone, or to a common side of two cones of SPM($x$); \\
\hspace*{0.3cm} (vi) For any point $u$ belonging to one side $\gamma(z, p)$ of a cone $\mathcal C(z; p, q)$ of SPM($x$), the angle $\angle_u (z, p)$ is at least $\pi$ inside $\mathcal C(z; p, q)$. \\
\end{proposition}

\begin{proof}
For the clarity of the paper, we put the proof in the Appendix \ref{app}.
\end{proof}


\begin{lemma}\label{sp}
For any point $y$ of $\mathcal K,$ the shortest path between $x$ and $y$ passes via the apex of the cone of SPM($x$) containing $y$.
\end{lemma}

\begin{proof}
By the Proposition \ref{pr-spm} (v), any point $y$ of $\mathcal K$ belongs to at least one cone of SPM($x$).
Let $\mathcal C(z; p, q)$ be a cone of SPM($x$) containing $y.$ Note that $z$ can coincide with $x$. \\
Suppose the contrary of the lemma's affirmation, i.e., $z\not \in \gamma(x, y).$
Let $s$ be a point of $\mathcal K$ such that $s \in \gamma(x, y) \cap \gamma(z, p)$ (the case where $s\in \gamma(x, y)\cap\gamma(z, q)$ is similar) and $s\neq z$.
By the definition of SPM($x$) (2), $\gamma(x, p) = \gamma(x, z)\cup\gamma(z, p)$ and for any point $r\in\gamma(z, p)$ the shortest path between $x$ and $r$ passes via $z$: $\gamma(x, r) = \gamma(x, z) \cup \gamma(z, r).$ Since $s \in \gamma(z, p),$ then $\gamma(x, s) = \gamma(x, z) \cup \gamma(z, s),$ which contradicts our assumption $z\not \in \gamma(x, y)$. \hfill $\Box$
\end{proof}

The next construction presented as a lemma, shows how to compute geodesic rays from a given geodesic segment inside $\mathcal K.$

\begin{lemma}\label{constr-demidroite}(Geodesic ray shooting)
Given a geodesic segment $\gamma (u, v) \not \subset \partial \mathcal K,$ it is always possible to construct two families of geodesic rays ($ R(u,v) $ of origin $u$ and of direction $v$ and $R(v, u)$ of origin $ v $ and of direction $ u $) containing the geodesic $\gamma (u, v)$.
\end{lemma}

\begin{proof}
We first show how to build a geodesic ray $ r(u, v)$ of origin $ u $ and of direction $ v.$ We can assume that the point $ v $ does not belong to $ \partial \mathcal K $, since otherwise the geodesic ray $ r(u, v) $ coincides with the $ \gamma (u, v)$ segment.

\begin{figure}[h]\centering
\includegraphics[width=0.8\textwidth]{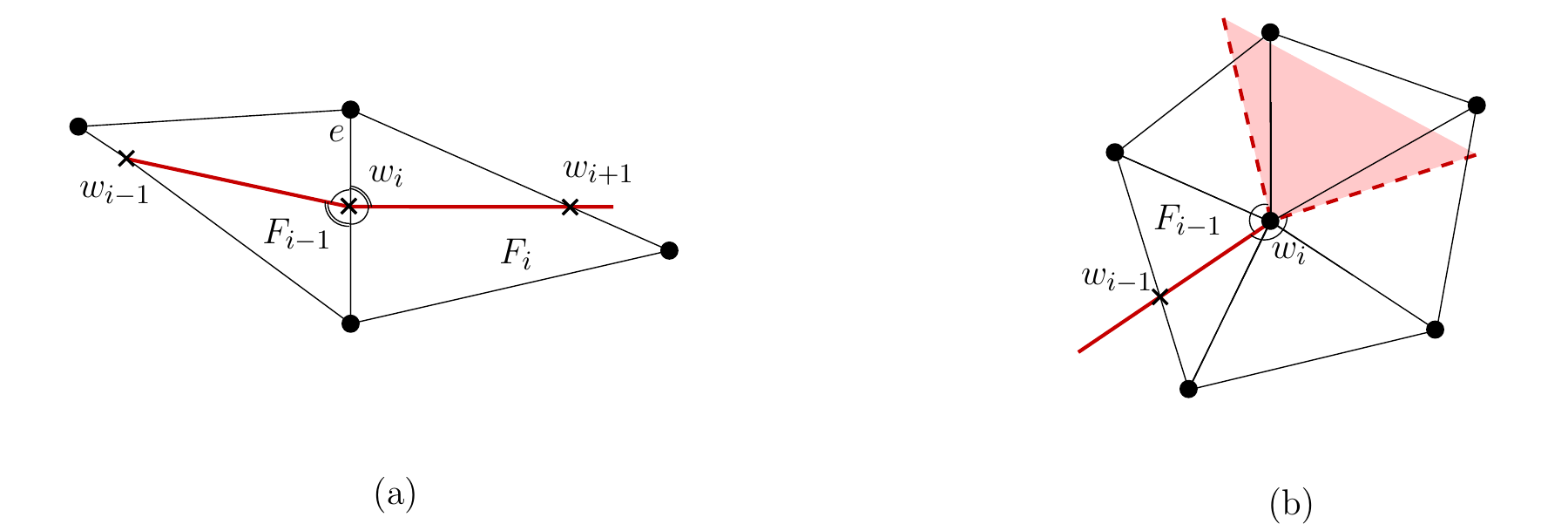}
\caption{Geodesic ray shooting.} \label{des-constr-dreptei}
\end{figure}

Let $w_1, w_2, \ldots,w_l $ $(l <k) $ be the intersection points of the geodesic $ \gamma (u, v) $ with the edges of $ \mathcal K $. We show further how to determine the points $w_{l+1},\ldots, w_k$ on $r(u,v).$ Let $F_l$ be the face of the complex containing the points $w_l$ and $ v,$ then $w_{l+1} \in \partial F_l$ is such that the resulting complementary angles $\angle_v (w_l, w_{l +1})$ are at least equal to $\pi $. \\

We distinguish two cases: case (a) where the point $w_i $, ($i \in \{l+1, \ldots, k \}$) is an inner point of an edge of $F_{i-1} $ and the case (b) where $w_i$ is a vertex of $ \mathcal K.$

In the first case, we suppose that $ w_i \in F_{i-1} $ (see Fig. \ref{des-constr-dreptei} (a)). We suppose that $ F_{i-1} $ contains the points $w_{i-1}$ and $w_i$ and that the face $F_i$ contains the common point $w_i$. Therefore, $ w_{i+1} \in \partial F_i $ is such that the complementary angles $ \angle_{w_i}(w_{i-1}, w_{i +1})$ are equal to $\pi.$

In the second case (see Fig. \ref{des-constr-dreptei} (b)), the point $w_{i+1}$ belongs to one of the faces incident to $w_i$, such that the complementary angles $\angle_{w_i}(w_{i-1},w_{i+1})$ are equal to $\theta(w_i)/2 \geq \pi $.\\
We claim that a face of $\mathcal K$ cannot be visited twice by a geodesic. We prove this by assuming the contrary.
Since each face of $\mathcal K$ is an Euclidean triangle, the geodesic intersects twice at least one of the edges of the face.
An edge of the complex is locally convex and so it is any geodesic in $\mathcal K$, thus they both are convex sets in $\mathcal K$. Since the intersection of two convex sets is convex, we obtain that the intersection of the geodesic and the edge, which is a set of two distinct points, is convex in $\mathcal K$, which is absurd.\\
Since our complex $ \mathcal K $ is finite and each face of $\mathcal K$ can be visited only once by a geodesic, after a finite number of steps, we necessarily reach a point $w_k$ belonging to the outer face of $ \mathcal K.$ By the choice of vertices $w_1, \ldots, w_k ,$ the point sequence $u, w_1, \dots, w_l, v, w_{l +1}, \ldots, w_k$ form a chain (polygonal line) locally convex and hence convex, of origin $ u $ and of direction $ v $. \hfill $\Box$
\end{proof}


\subsection{The sweep of the complex}

In this sub-section, given a CAT(0) planar complex $\mathcal K$ and a point $x\in\mathcal K$, we propose a sweeping-line algorithm for constructing the shortest path map in a CAT(0) planar complex $\mathcal K$.
The algorithm visits the faces of $\mathcal K$ using a \emph{sweeping(or level)-line}, denoted by $C(r).$ We call \emph{sweep events} all the crossing points of $C(r)$ from one face to another. The sweeping line $C(r)$ at time $r$ consists of a sequence of arcs of circles of radius $r$ and a fixed center $x$, concatenated by \emph{break points}. These arcs appear when $C(r)$ crosses the sweep events.\\
\noindent When $C(r)$ passes via an event $h$, we construct the geodesic $\gamma(x, h)$ between $h$ to the root $x$ in $\mathcal K$.
All the geodesics constructed this way form a partition of the complex $\mathcal K$ into convex regions called \emph{cones} of SPM($x$).\\
There are two types of events encountered by the sweeping line of $\mathcal K$: edge-events and vertex-events.
\begin{figure}[h]\centering
   \includegraphics[width=0.75\textwidth]{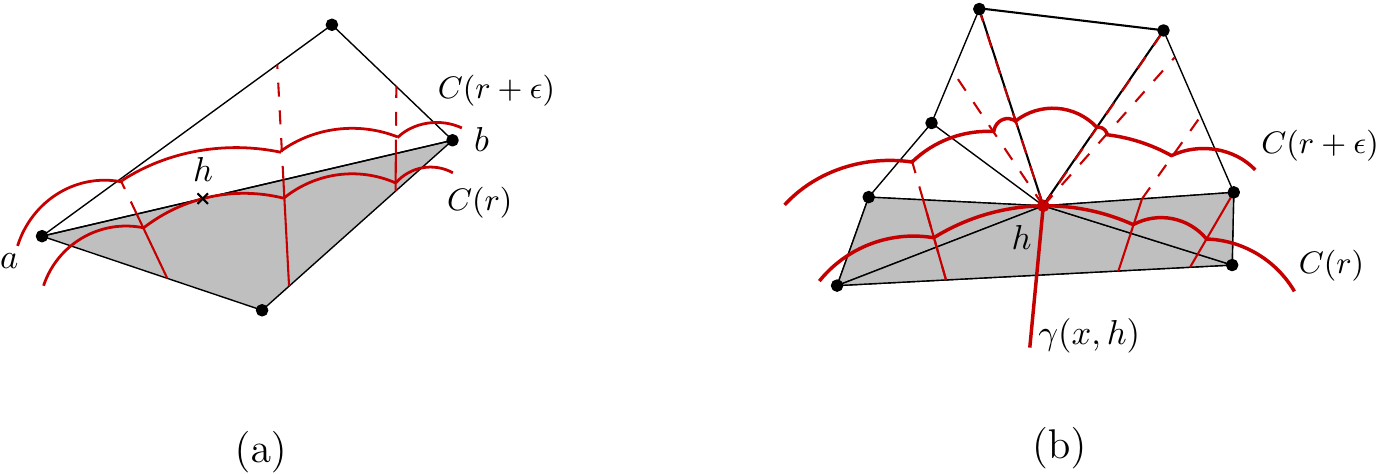}
 \caption{(a) Edge-event, (b) Vertex-event.} \label{des_general1}
\end{figure}

The \emph{vertex-events} are all vertices of $\mathcal K.$ The crossing of the sweeping line $C(r)$ via an event $h$ of this type creates new arcs on $C(r)$ (see Fig. \ref{des_general1} (b)). When $C(r)$ passes via $h$, we construct the geodesic $\gamma (x,h)$ connecting the points $x$ and $h$. The information for the sweep of $\mathcal K$ is then transmitted from the face already covered by $C(r)$ which contains $h$ to all other faces incident to $h$.

The \emph{edge-events} are inner points of edges of $\mathcal K.$ For some edge-event $h,$ there exists a face $F = \Delta(a, b, c)$ of $\mathcal K,$ and a cone $\mathcal C (z; p, q)$ of SPM($x$) such that $h$ is the closest point of $F$ to $x$ in $\mathcal C(z; p, q).$
An edge $e$ of $\mathcal K$ can contain multiple edge-events (see Fig. \ref{des_general1} (a)). The crossing of the sweeping line $C(r)$ via some edge-event $h$ does not change the shape of $C(r)$. These events are only used to transmit the information for the sweep from one face to another.\\

Let $\mathcal C(z; p, q)$ be a cone of SPM($x$) and $v$ a point on a side of $\mathcal C(z; p, q).$ Suppose $v$ belongs to $\gamma(z, p).$
By the definition of the shortest path map, the inner angle $\angle_v(z, p)$ of the cone $\mathcal C(z; p, q)$ is at least equal to $\pi$. The algorithm that we present in the next section builds the cones of SPM($x$) such that $\angle_v(z, p)$ are equal to $\pi$.

\begin{lemma}\label{plonj-con}
All cones of SPM($x$) can be embedded in the plane as acute triangles.
\end{lemma}

\begin{proof}
Let $ \mathcal C(z; p, q)$ be a cone of SPM($x$) in $\mathcal K.$
By Lemma \ref{pr-spm} (i), int($\mathcal C(z; p, q)$) contains no vertex of $\mathcal K.$ Therefore, int($\mathcal C(z; p, q)$) contains no points of negative curvature.
By Lemma \ref{pr-spm} (ii), the inner angle of origin $z$ is less than $\pi$. By the remark preceding the lemma, for every point $v$ situated on a side of the cone $\mathcal C(z; p, q)$ the angle of origin $v$ is equal to $\pi$.
By the definition of SPM($x$), for each vertex $v$ of $\mathcal K$ (including the vertices on the boundary of the complex), SPM($x$) contains the geodesic $\gamma(x,v)$ contained in a side of at least one cone of SPM($x$).
The segment $\gamma(p,q)$ cannot contain any vertices, as in the opposite case the cone $\mathcal C(z; p, q)$ is divided in at least two cones of SPM($x$).
Thus $\gamma(p,q)$ belongs to an edge of the boundary of $\mathcal K.$ Therefore, for any point $u \in \gamma(p, q)$, the angle $\angle_u(p, q)$ is equal to $\pi$.
In summary, this implies that $\partial \mathcal C(z; p, q)$ is a geodesic triangle in $\mathcal K.$ Thus there exists a comparison triangle $\Delta(z', p', q')$ in $\mathbb{E}^2$, which represents the unfolding of $\mathcal C(z; p, q)$ in the plane.
\hfill $\Box$
\end{proof}

Subsequently, we associate to each edge-event $h$ of SPM($x$) the distance $d(x, h)$. If $h$ is a vertex-event of SPM($x$), we associate to $h$ the distance $d(x, h)$ and the two formed angles between the geodesic $\gamma(x,h)$ and the edges of $\mathcal K$ incident to $h$ in the face covered by $C(r)$ at time $r$.

\medskip
Let $F =\Delta(a, b, c)$ be a face of $\mathcal K$ and $h$ an event of $F$. If $h$ is an edge-event of $F$, knowing the distance $r = d(x, F)$ one can construct inside $F$ a unique arc of the circle $C(x, r)$.
If $h$ is a vertex-event, knowing the distance $r = d (x, F)$ and the two angles between $\gamma(x, h)$ and the edges of $F$ incidents to $h,$ it is possible to build inside $F$ a unique arc of the circle $C(x, r)$.

\begin{lemma}\label{constr-cerc}
Let $h$ be an event contained in a face $ F = \Delta(a, b, c)$ and let $\mathcal C (z; p, q)$ be a cone of SPM($x$) which intersects $F$ and contains $h$.
If $h$ is an edge-event $F$, using the distance $r = d(x, F)$ one can construct inside $F \cap \mathcal C(z; p, q)$ a single arc of the circle $C(x, r)$.\\
If $h$ is a vertex-event, using the distance $r = d(x, F)$ and the two angles between $ \gamma(x, h)$ and the two edges of $F$ incident to the vertex $h,$ one can construct inside $ F \cap \mathcal C(z; p, q)$ a single arc of the circle $C(x, r)$.
\end{lemma}

\begin{proof}
Suppose first that $h$ is an inner point of edge  $ab$ of $F$.
Since $h$ is an event, it is the point of $F$ closest to $x.$ Then the edge $ab$ is tangent to the circle $C(x, r)$ in $\mathcal K$, where $r = d(x,F)$. We build in the Euclidean plane the isometric image of the triangle $ \Delta(a, b, c)$, which we denote by $\Delta (a ', b', c ').$ Let $h'$ be the point of the side $a'b'$ so that $d_{\mathbb{E}^2}(h', a') = d (h, a)$ and $d_{\mathbb{E}^2}(h', b') = d(h, b).$ We will now build the image of the point $x$ in $\Delta (a ', b', c ')$.
For this, it suffices to construct a segment $[x'h']$ outside $\Delta(a ', b', c ')$ of length $r$, with one end in $h' \in [a'b']$ and perpendicular to the segment $[a'b']$. The locus of all points $x'$ of the plane, represents a single point. Therefore, it is possible to build a single circle of radius $r$ whose center is the other endpoint of the segment $[x'h']$.

Consider now the case where $h$ is a vertex-event of $F$. Suppose $h = c$, then $ha$ and $hb$ are two edges of $F$ incident to $h$. We construct the isometric image of the triangle $F = \Delta(a, b , h)$ in the plane, which we denote by $\Delta(a ', b', h').$ We seek to locate the image of $x$ in the plane using the following: $h$ is the point of $F$ closest to $x$ and the distance $d(x, h)$ equals to $r.$ The locus of all points $x'$ of the plane located at a distance $r$ of a fixed point $h'$ is a circular arc. However, knowing the angles formed between the edge $ha$ and the segment $xh$ and between $hb$ and $xh,$ the locus of all points $x'$ plan verifying these conditions is a single point. Thus, we can construct a single circle $C(r)$ in $F$ (centered in $x$ and of radius $r$).
\hfill $\Box$
\end{proof}

The following lemma shows how to identify new events in a face of the complex during the sweep of $\mathcal K.$
\begin{lemma}\label{constr-gamma}
Let $F =\Delta(a, b, c)$ be a face of $\mathcal K$, $\mathcal C(z; p, q)$ a cone of SPM($x$) crossing $F$ and $h$ an event of the sweep inside $F \cap\mathcal C(z; p, q)$. It is possible to determine a new event $h_i$ in $F \cap \mathcal C(z; p, q)$ and the distance $d(x, h_i)$ to $x$. Moreover, if $h_i$ is a vertex of $F$ then one can construct the geodesic $\gamma(x, h_i)$ in $\mathcal K.$
\end{lemma}

\begin{proof}
By the previous lemma, using the information associated to the event $h$, it is possible to build inside $F$ an arc $C(x,r).$ Suppose that for $r^*$ the arc $C(x,r^*)$ is maximal by inclusion in $F$. Then $C(x,r^*)$ intersects at least one edge $F$ in a point $h_i$. As the events are crossings points from one face to the other, $h_i$ is an event of $C(r).$ We will build the geodesic segment $\gamma(x, h_i)$ in $\mathcal K$ and calculate the distance $d(x,h_i). $ Since the cone $\mathcal C(z; p, q)$ of SPM($x$) contains events $h$ and $h_i$ then $d(x, h_i) = d(x, z) + d(z,h_i)$ and $\gamma(x, h_i) = \gamma(x, z) \cup \gamma(z, h_i).$ Therefore, to construct $\gamma(x, h_i)$ in $\mathcal K$ and calculate the distance $d(x, h_i)$ it suffices to construct the geodesic segment $\gamma(z, h_i)$ in $\mathcal K$ and calculate the distance $d(z,h_i)$.

For this, we use the same reasoning as in the proof of the previous lemma.
Let $\Delta(a', b', c')$ be the isometric image of the triangle $\Delta(a, b, c)$ in the plane and $h', h'_i \in \Delta(a', b', c')$ the respective images of $h$ and $h_i$ in $\Delta(a, b, c).$ By Lemma \ref{plonj-con}, we know
that the cone $\mathcal C (z; p, q)$ can be unfolded in the plane in the form of an acute triangle. Therefore, we can construct in the plane the isometric image of all points of the cone.
Let us analyze separately the case where $h$ is an edge-event and the case where $h$ is a vertex-event.

If $h$ is an edge-event $ab$ (see Fig. \ref{des_gamma} (a)), by the previous lemma the distance $d(z, h)$ is known. Its suffices to build the point $z'$ in the plane, such that $z\not\in\Delta(a', b', c') $, \
$d_{\mathbb{E}^2}(z', h') = d(z,h)$ and $[z'h]$ is perpendicular to the side $[a'b'].$ Let $s'$ be the intersection point of the segment $[z'h'_i] $ with the edge $[a'b'].$ We denote by $s$ the inner point of $ab$ of $F$ such that $d(a, s) = d_{\mathbb{E}^2}(a, s')$ and $d(s, b) = d_{\mathbb{E}^2}(s', b').$ In order to construct the geodesic $\gamma(z, h_i)$ in $\mathcal K$ it suffices to launch the geodesic ray $r(h_i,s)$ towards $z$.
\begin{figure}[h]\centering
   \includegraphics[width=0.75\textwidth]{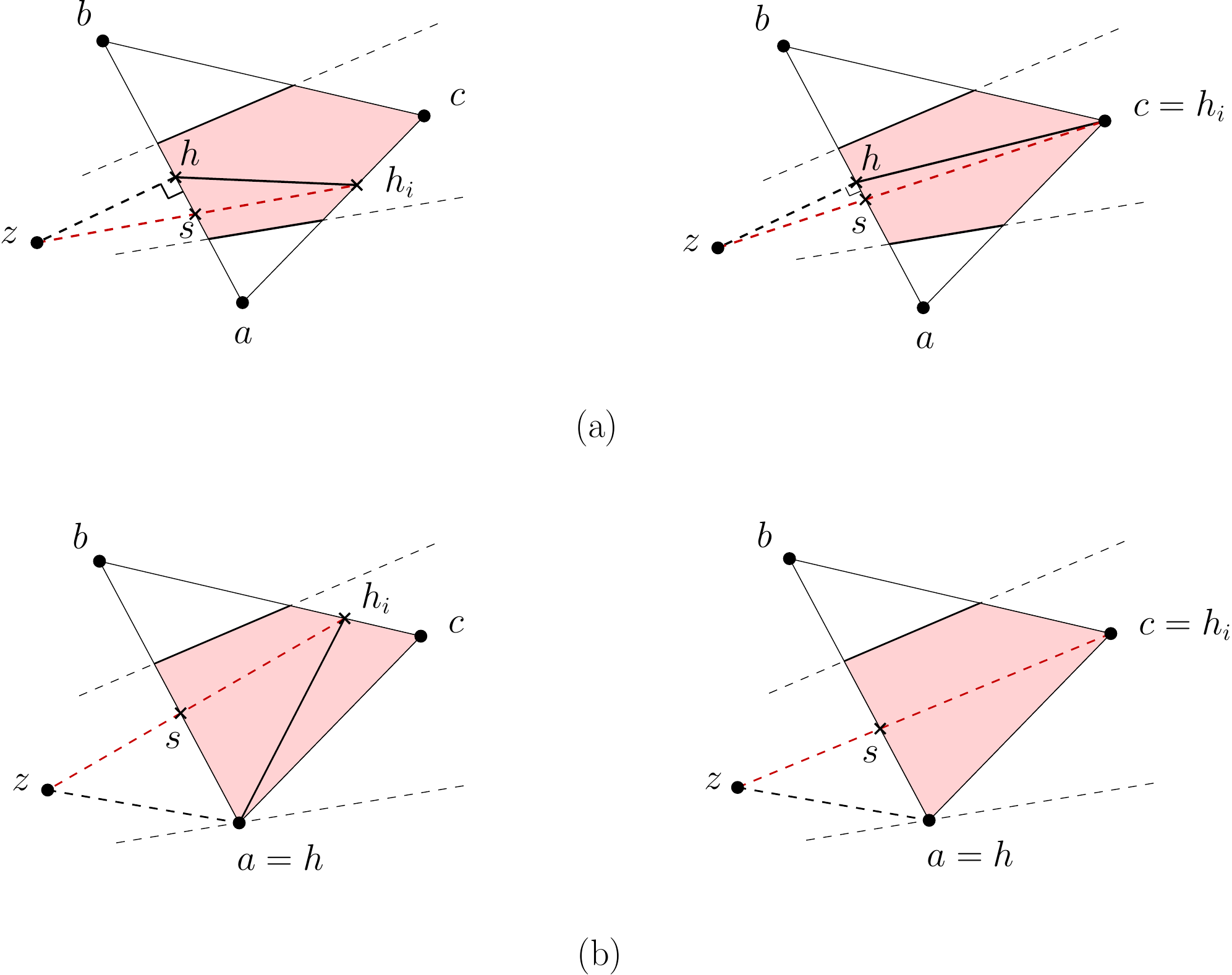}
 \caption{Computing new events inside a face of $\mathcal K$.} \label{des_gamma}
\end{figure}

If $h=a$ is a vertex-event of $F$ (see Fig. \ref{des_gamma} (b)), by the previous lemma, the following measurements are known: distance $d(z, h)$ and angles $\alpha_1,$ $\alpha_2$ between the geodesic $\gamma(z, h)$ and
the edges incident to $h$ in $F$. It suffices to build the point $z'$ in the plane such that $z\not\in\Delta(a', b', c')$, $d_{\mathbb{E}^2}(z', h') = d(z,h)$ and $[z'h]$ form the angles $\alpha_1,$ $\alpha_2$ with the edges
$[a'b]$ and $[a'c]$. Let $s'$ be the intersection point of the segment $[z'h'_i]$ with the edge $[a'b']$ in the plane. The case where $[z'h'_i]$ intersects the edge $[a'c]$ is similar.

Let $s$ be the point of the edge $ab$, so that $d(a,s) = d_{\mathbb{E}^2}(a, s')$ and $d(s, b) = d_{\mathbb{E}^2}(s', b').$ In order to construct the geodesic $\gamma(z,h_i)$ in $\mathcal K$ it suffices to launch the geodesic ray $r(h_i, s)$ towards $z$. \hfill $\Box$
\end{proof}

\subsection{Computing SPM($x$)}

In this section we present an efficient algorithm which constructs the shortest path map SPM($x$) in a CAT(0) planar complex $\mathcal K$ with $n$ vertices using a data structure of size $O(n^2)$. This algorithm traverses the faces of the complex using a sweeping line from a given source-point $x \in \mathcal K$ and builds simultaneously the cones of SPM($x$).

\subsubsection{Data structure of the algorithm}

The data structure used by our algorithm, consists of two substructures: a \emph{static} substructure $D_s,$ which does not change during the steps of the algorithm, and a \emph{dynamic} substructure $D_d,$ which is initialized at step one of the algorithm and changes during the sweep of $\mathcal K$.\\
\noindent The static substructure contains the planar map of the complex $\mathcal K$ and the circular lists of angles of every vertex of $\mathcal K$. At time $r$ of the sweep, the dynamic substructure contains a priority queue $\mathcal Q$ of events crossed by $C(r)$ and a list $\mathcal C$ of cones constructed up to $C(r).$ Note that the intersection of a triangular face of $\mathcal K$ with a cone of SPM($x$) is at most a quadrilateral. Thus, the dynamic substructure $D_d$ contains, for any face $F$ of $\mathcal K$ the list of quadrilaterals coming from the intersection of $F$ with cones of $\mathcal C.$\\
\noindent We use the representation of the complex as a planar map \cite{BeChKrOv} which is a doubly-connected edge list. This representation allows us to use a data space of linear size with respect to the number of vertices of $\mathcal K.$

\begin{figure}[h]\centering
\includegraphics[width=0.8\textwidth]{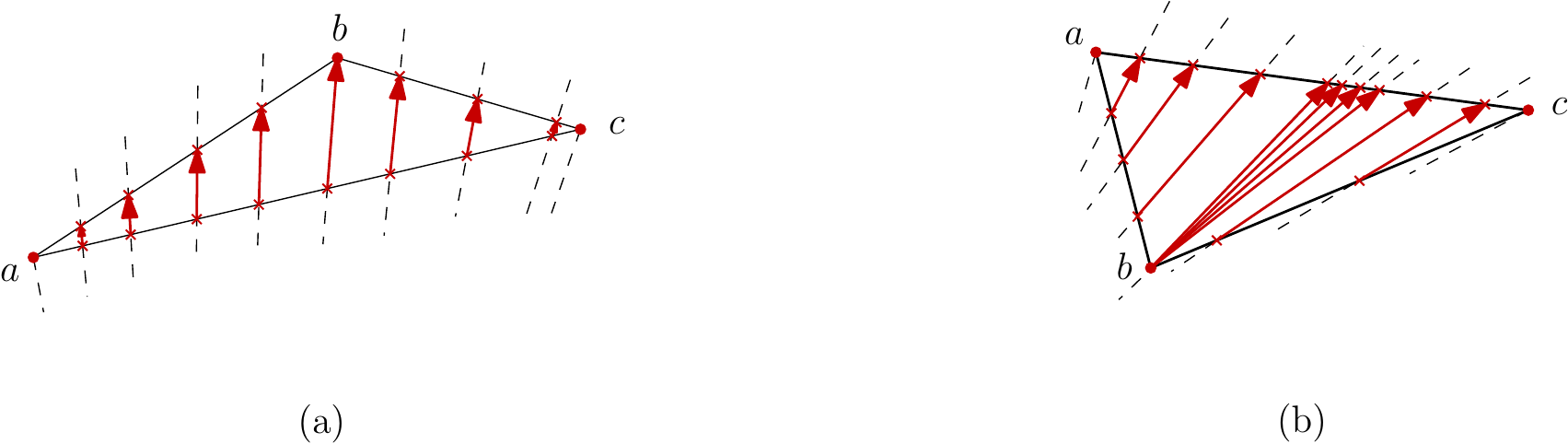}
\caption{Intersections of the faces of $\mathcal K$ with cones of SPM($x$).}\label{intersections1}
\end{figure}
Given a CAT(0) planar complex $\mathcal K$ and a point $x \in {\mathcal K}$ we construct the shortest path map SPM($x$) as a geodesic tree structure.

\subsubsection{Algorithm}

Given a CAT(0) planar complex ${\mathcal K}$ and a source-point $x \in {\mathcal K} $ we present an algorithm which computes the shortest path map SPM($x$) in $\mathcal K$.

First we assume that $x$ is an inner point of a face $F = \Delta(a, b, c)$ of $\mathcal K.$ In this case, we add $x$ to set of vertices of $\mathcal K$, $V(\mathcal K): = V (\mathcal K) \cup \{x \}$ and the segments $[xa], [xb] $ and $ [xc] $ constructed inside $ F $ to the set of edges of $ \mathcal K $. Thus $ F $ is divided into three triangular faces $ \Delta (x, a, b),$ $ \Delta (x, a, c)$ and $ \Delta (x, b, c). $

From the root point $x$ of $\mathcal K $ the algorithm traverses all the faces of the complex using a sweep line $C(r) $ and passes from a face to another by the sweep events. These points are crossed by the sweeping line in ascending order of their distances to the root $x.$ We use a priority queue $\mathcal Q$ to store the events encountered by $C(r)$.

The event at the top of $\mathcal Q$ is extracted from the priority queue, and using Lemma \ref{constr-cerc} the sweeping line is then constructed in the faces incident to this event. Moreover, we determine in these faces new events which are included in $ \mathcal Q $ according to their distances to $ x $.
When the sweeping line crosses a vertex-event $z$ of $ \mathcal K $, the algorithm builds the geodesic segment $\gamma(x, z)$.

Let $p$ and $q$ be two vertices on the sweeping line $C(r)$ at time $r$. In this case, $ p $ and $ q $ are equidistant from $x $. Let $ \gamma (x, p) $ and $ \gamma (x, q) $ be two geodesics constructed during the sweep of the complex and let $z$ be a vertex of $ \mathcal K $ such that $ z \in \gamma (x, p) \cap \gamma (x, q) $ and $ \angle_z (p, q) \leq \pi.$ The set of points between the geodesics $ \gamma (z, p) $ and $ \gamma (z, q) $ defines a \emph{partial cone} of SPM($x$) up to the sweeping line, which is denoted by $\mathcal C (z; p, q).$ We initialize the list $\mathcal C$ which registers all the partial cones formed between two consecutive geodesics built up to the sweeping line.
For any face $ F $ of $ \mathcal K $, we use a list $ \mathcal L(F) $ containing quadrilaterals obtained from the intersection of $ F $ with partial cones of $ \mathcal C $ (see Fig. \ref{intersections1}). The quadrilaterals in $ \mathcal L (F) $ are sorted by the coordinates in $ F $ of the points of intersection of edges of $ F $ with the sides of partial cones of $ \mathcal C $.

\noindent We describe now the steps of the algorithm in a more detailed way.\\

\noindent\textbf{Initialization step.} Given the root point $ x $ in $ \mathcal K ,$ the algorithm determines the events in all the faces incident to $ x $ and includes them in the priority queue $\mathcal Q $ as follows: \\
Let $ F = \Delta (x, a, b) $ be a triangular face of $ \mathcal K $ incident to $ x $. An edge-event $ h\in F$ is an inner point of the edge $ab $, such that the segment $[xh]$ is perpendicular to $ab$ in $ F. $ In other words, a point $ h \in  ab$ is an edge-event of $ F $ if the arc $ C(x, r) $ is maximal by inclusion inside $ F.$ 
If such an event $ h $ exists, we associate to this event the distance $ d(x, h). $ Note that the distance $d(x, h)$ is calculated using the Euclidean metric inside the face $ F $ of $ \mathcal K.$

The vertices $a, b$ of $F$ are vertex-events of $F$. We associate to each vertex-event $ h $ of $ F $, the distance $ d (x, h) $ together with the two angles formed between the edge $xh$ and the two edges incident to $ h $ in $ F.$\\
The priority queue $ \mathcal Q $ is initialized by inserting in $ \mathcal Q $ all the edge-events and vertex-events of faces incident to the root point $ x $ according to their distances to $ x $. The top event of $ \mathcal Q $ is the closest event to $ x $ belonging to a face containing $ x $.
The list of partial cones $ \mathcal C $ is initialized as the set of cones $ \mathcal C (x; a, b) $, where $ F = \Delta (x, a, b) $ is a face incident to $x.$
At this step, for each face $ F $ incident to $ x, $ the list $ \mathcal L (F) $ is initialized with the face $ F $ which can be seen as a degenerated quadrilateral.


\noindent\textbf{Step $k$.} After $k-1$ steps, let $ \mathcal Q $ be the priority queue of all events crossed by $C(r)$ and let $ \mathcal C $ be the list of partial cones built up to the sweeping line. Let $h$ be the top event of the priority queue $ \mathcal Q.$ Then $h$ is extracted from $\mathcal Q$ and $h$ is treated among one of the two following cases:\\

\noindent\textbf{Case 1}: $h$ is an edge-event. Let $ab$ be the edge of $\mathcal K$ containing $h$ and let $ F = \Delta (a, b, c)$ be the triangular face of $ \mathcal K $, such that $ c $ was already crossed by the sweeping line. We denote by $ F' = \Delta (a, b, c')$ the adjacent face of $ F $ sharing a common edge $ab$.

\noindent Using local coordinates of $ h $ on the edge $ab$ and the list $ \mathcal L(F) $ we can determine by a binary search the quadrilateral $Q$ of $ \mathcal L(F)$ containing $h$ as follows. We identify the quadrilateral $Q$ of $\mathcal L(F)$ such that $h$ is located at the left from a side of $Q$ and at the right from another side of $Q.$
Then we determine the partial cone $ \mathcal C(z; p, q)$ of $ \mathcal C $ containing this quadrilateral $Q.$

\noindent Since the geodesics $\gamma(x, p)$ and $\gamma(x, q)$ are constructed in the previous steps, using Lemma \ref{constr-demidroite} we are shooting the following geodesic rays: $r(x, p) $ and $ r(x, q)$ inside the face $F'.$ The cone $ \mathcal C (z; p, q)$ is replaced in $ \mathcal C $ by the cone $ \mathcal C (z; p', q'),$ where $ \gamma (x, p')$ and $\gamma(x, q')$ are the previously constructed geodesic.

Using the method described in Lemma \ref{constr-cerc} for the case where $h$ is an inner point of an edge, we identify all the new events $h_i $ inside the quadrilateral $F \cap \mathcal C (z; p', q').$ By Lemma \ref{constr-gamma}, for each new event $h_i$, we compute the distance $d(x, h_i)$ knowing the distances $d(x,h)$ and $d(h, h_i)$.

\noindent For each new vertex-event $h_i$ of $F \cap \mathcal C(z; p', q') $, we construct the geodesic $\gamma(x, h_i)$.
We know that $\gamma(x, h_i) = \gamma(x, z) \cup \gamma(z, h_i).$ As $z$ is a vertex of $ \mathcal K $ already treated, then $ \gamma (x, z) $ is constructed at this step. We use Lemma \ref{constr-gamma} to build the geodesic $ \gamma(z, h_i) $.

\noindent We replace in $\mathcal C$ the cone $ \mathcal C (z; p', q')$ containing the event $h$ by new partial cones $ \mathcal C (z; p', h_i) $ and $ \mathcal C (z; h_i, q'). $ We then update the list $\mathcal L(F')$ for the face $F'.$

\noindent For each new edge-event or vertex-event $h_i$ of $F,$ if $h_i$ does not belong to $ \mathcal Q $ we insert $h_i$ in the priority queue $\mathcal Q $ among its distance $d(x, h_i)$.

\textbf{Case 2}: $h$ is a vertex-event. Let $\mathcal C (z; p, h)$ and $\mathcal C (y; h, q)$ be the two partial cones of $\mathcal C$ sharing the common side $ \gamma (z, h),$ such that $z \in \gamma(y, h).$ We denote by $F_j$ (with $j> 1$), the faces incident to $ h $ crossed by the sweeping line and $F'_j$ (with $j> 1$), the faces incident to $ h $ but not crossed yet by $C(r)$.

Knowing the measurements of the angles of origin $h$ between the geodesic $ \gamma(z, h) $ built previously and the incident edges of $h$ we can calculate in linear time with respect to $deg(h)$ the value of $ \theta(h)$ with the origin in $h.$
If $\theta(h) = 2 \pi,$ using Lemma \ref{constr-demidroite}, we launch a geodesic ray $\gamma (z, h)$ inside the faces incident to $h$. We replace in $ \mathcal C $ the initial cones $ \mathcal C(z;p,h)$ and $ \mathcal C (y; h, q) $ by the cones $ \mathcal C (z; p', h')$ and $\mathcal C (y; h', q'),$ where $\gamma(z, p'),$ $\gamma(y, q')$ and $\gamma(z, h')$ are the geodesics constructed inside faces $F'_j$.
\begin{figure}[h]\centering
   \includegraphics[width=0.7\textwidth]{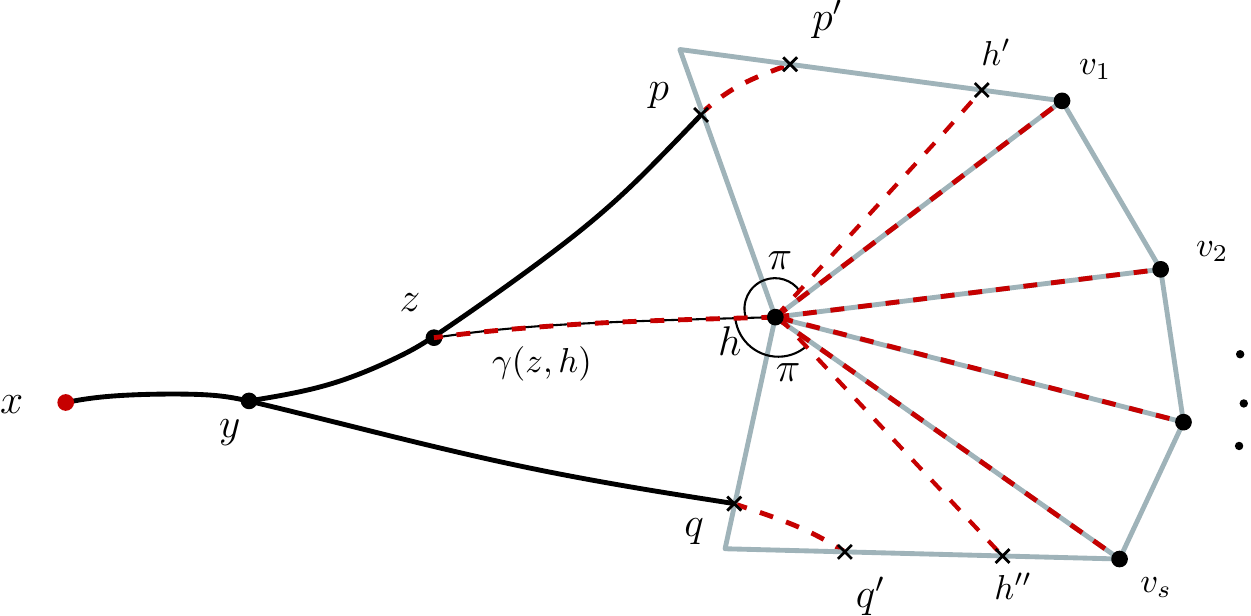}
 \caption{The case where $\theta(h)>2\pi$.} \label{des_general}
\end{figure}

Let $h$ be an negative curvature vertex of $\mathcal K$, such that $2l\pi <\theta(h)<2(l +1)\pi$ where $l\geq 1$ (see Fig. \ref{des_general}) . In this case, we launch a geodesic ray $r(z, h)$ inside the faces $F'_j$ by constructing two geodesics $\gamma(h, h')$ and $\gamma (h, h'')$ such that the angles $\angle_h (z, h')$ and $\angle_h (z, h'')$ are equal to $\pi$. We replace in $\mathcal C$ the cones $\mathcal C(z; p, h)$ and $\mathcal C(y; h, q)$ by the cones $\mathcal C(z;p', h')$ and $\mathcal C(y; h'', q').$ For each face $F'_j$, we update the list of quadrilaterals $\mathcal L(F'_j).$

As $\theta(h)>2l\pi,$ the angle $\angle_h(h', h'')$ is greater than $2(l-1)\pi$. Therefore, there exists at least one vertex $v$ of $F'_j$ such that $\angle_h (z, v)> \pi.$ Let $s$ be the number of vertices $v_t$  of $F'_j$ adjacent to $h$, such as $ \angle_h(z,v_t)> \pi$.
To ensure that every angle corresponding to a cone of apex $h$ is less than $\pi,$ we construct exactly $s$ geodesics between $h$ and the vertices $v_t,$ $t = 1,\ldots,s$ of $F'_j.$ Since the angles inside a face of $\mathcal K$ are strictly less than $\pi,$ $s$ is such that $l\leq s \leq deg(h)$.
We can then update the list of partial cones $\mathcal C$ by adding the partial cones $\mathcal C(h; v_t, v_{t +1}),$ $ 1\leq t \leq s-1$ of apex $h$ and the sides $\gamma (h, v_t), \gamma(h, v_{t +1})$.
Moreover, the cones $\mathcal C (z; p, h)$ and $\mathcal C (y; h, q)$ are replaced in $\mathcal C$ by the cones $\mathcal C (z; p', v_1)$ and $\mathcal C (y; v_s, q'),$ where $p', q'$ belong to the boundary of $F'_j$.
For each face $F'_j$ ($j> 1$), we update the lists $\mathcal L(F'_j)$ of quadrilaterals obtained as intersections of partial cones with $F'_j.$
Using Lemma \ref{constr-cerc} (where $h$ is a vertex of $\mathcal K$), we determine all the new events $h_i$ by constructing circular arcs $ C(z, d(z, h) + r)$ inside $F'_j \cap \mathcal C(z; p', v_1)$, constructing circular arcs $C(y, d (y, h) + r)$ inside $F'_j \cap \mathcal C (y; v_s, q')$ and constructing circular arcs $C(h, d (h, r))$ inside $ F'_j \cap \mathcal C(h; v_t, v_{t +1}) $ for $t = 1, \ldots, s-1$.

To each event $h_i$ of a cone $\mathcal C(z; p', v_1)$ or $\mathcal C(y;v_s,q)$ we associate the distance $d(x, h_i)$. Knowing the distances $d(x,h)$ and $d(h,h_i),$ by Lemma \ref{constr-gamma} we can calculate the distance $d(x,h_i)$.
To each event $h_i$ of the cone $\mathcal C(h; v_t, v_{t +1})$ (for $t = 1,\ldots,s-1$), we associate the distance $d(x,h_i) = d(x,h) + d(h,h_i)$ where $d(h, h_i)$ is calculated using the Euclidean metric inside the triangular face containing both $h$ and $h_i$.

For each vertex-event $h_i$, we construct the geodesic connecting $h_i$ and $x$ as follows. If $h_i$ is a vertex of one of the cones $\mathcal C(z; p',v_1)$ or $ \mathcal C(y;v_s,q'),$ then using Lemma \ref{constr-gamma}, we construct the geodesic $\gamma(y,h_i).$ We associate to $ h_i$ the two angles between $\gamma(y, h_i)$ and the edges incident to $h_i$ in $F'_j$. \\
If $h_i$ is a vertex inside a cone $\mathcal C(h; v_t, v_{t +1})$ ($1\leq t \leq s-1$), then $h_i$ coincides with one of vertices $v_t$ or $v_{t+1}.$ Since the geodesics $\gamma(h, v_t)$ and $\gamma(h, v_{t +1})$ are edges of $\mathcal K$, it remains to associate to the event $h_i$ the two angles between $\gamma(h,h_i)$ and the edges incident to $h_i$ in $F'_j.$

We insert new events $h_i$ in the priority queue $\mathcal Q$ according to their distances $d(x, h_i)$.
The $k$ step is repeated until all the faces of the complex are covered by the sweeping line. \\
We give below a brief and informal description of the algorithm.

\begin{center}
\framebox{
\parbox{15cm}{
\vspace{0.05cm}
\noindent{\bf Algorithm} {\sc Computing SPM($x$)}\\
{\footnotesize
 \noindent {\bf Input:} a CAT(0) planar complex ${\mathcal K},$ a point $x\in {\mathcal K}$ and a data structure $D_s$\\
 \noindent {\bf Output:} Shortest Path Map SPM($x$) of root $x$\\
\rule{\linewidth}{.7pt}
\textbf{Initial Step} \\
 \hspace*{0.1cm} \textbf{for} every face $F$ incident to $x$ \textbf{do} \\
 \hspace*{0.7cm} find all the events $h_i$ and the distances $d(x,h_i)$ \\
 \hspace*{0.7cm} initialize the list $\mathcal L(F)$ \\
 \hspace*{0.1cm} \textbf{end} \\
 \hspace*{0.1cm} insert the events $h_i$ in the priority queue $\mathcal Q$ according to $d(x,h_i)$ \\
 \hspace*{0.1cm} initialize the list of partial cones $\mathcal C$ \\
\textbf{Iterative Step} \\
 \hspace*{0.1cm} \textbf{while} $\mathcal Q \neq \emptyset$ \textbf{do}\\
 \hspace*{0.7cm} extract the top event of $\mathcal Q$ \\
 \hspace*{0.7cm} \textbf{if} $h$ is a vertex of $\mathcal K$ \textbf{then}\\
 \hspace*{1.4cm} determine all the cones $\mathcal C(z;p,h), \mathcal C(y;h,q)\in \mathcal C$ which contain $h$ \\
 \hspace*{1.4cm} {\sc Vertex-Event($h,$ $\mathcal C(z;p,h), \mathcal C(y;h,q)$)} \\
 \hspace*{0.7cm} \textbf{else} \\
 \hspace*{1.4cm} determine the partial cone $\mathcal C(z;p,q)\in \mathcal C$ containing $h$ \\
 \hspace*{1.4cm} {\sc Edge-Event($h$, $\mathcal C(z;p,q)$)} \\
 \hspace*{0.7cm} \textbf{end} \\
 \hspace*{0.1cm} \textbf{end} \\
 \hspace*{0.1cm} return SPM($x$) as the dynamic data substructure
}}}
\end{center}
\smallskip

\begin{center}
\framebox{
\parbox{15cm}{
\vspace{0.05cm}
\noindent{\sc Vertex-Event($h,$ $\mathcal C(z;p,h), \mathcal C(y;h,q)$)} \\
{\footnotesize
 \hspace*{0.1cm} calculate $\theta(h)$  \\
 \hspace*{0.1cm} \textbf{if} $\theta(h)=2\pi$ \textbf{then} \\
 \hspace*{0.8cm} launch a geodesic ray $r(z,h)$ inside one of the faces incident to $h$  \\
 \hspace*{0.1cm} \textbf{else} \\
 \hspace*{0.8cm} launch the geodesic rays $r(z,h)$ inside faces incident to $h,$ where $2l\pi<\displaystyle\theta(h)<2(l+1)\pi$  \\
 \hspace*{0.1cm} \textbf{end} \\
 \hspace*{0.1cm} update the list of partial cones $\mathcal C$ \\ 
 \hspace*{0.1cm} \textbf{for every} face $F'_j$ incident to $h$ \textbf{do}\\
 \hspace*{0.8cm}  update the list $\mathcal L(F'_j)$ \\
 \hspace*{0.8cm} \textbf{for every} cone $\mathcal C(i;j,j')$ such that $h\in\mathcal C(i;j,j')$ and $\mathcal C(i;j,j')\cap F'_j\neq\emptyset$ \textbf{do}\\
 \hspace*{1.6cm} {\sc Find-New-Event($h,$ $F'_j,$ $\mathcal C(i;j,j')$)} \\
 \hspace*{0.8cm} \textbf{end} \\
 \hspace*{0.1cm} \textbf{end} \\
}}}
\end{center}
\smallskip

\begin{center}
\framebox{
\parbox{15cm}{
\vspace{0.05cm}
\noindent{\sc Edge-Event($h$, $\mathcal C(z;p,q)$)} \\
{\footnotesize
 \hspace*{0.1cm} inside the face $F'$ incident to $h,$ extend the sides of $\mathcal C(z;p,q)$ \\
 \hspace*{0.1cm} update the list of partial cones $\mathcal C$ \\
 \hspace*{0.1cm} update the list $\mathcal L(F)$ \\
 \hspace*{0.1cm} {\sc Find-New-Event($h,$ $F',$ $\mathcal C(z;p,q)$)} \\
}}}
\end{center}
\smallskip

\begin{center}
\framebox{
\parbox{15cm}{
\vspace{0.05cm}
\noindent{\sc Find-New-Event($h,$ $F,$ $\mathcal C(z;p,q)$)} \\
{\footnotesize
 \hspace*{0.1cm} find inside $F\cap \mathcal C(z;p,q)$ new events $h_i$ \\
 \hspace*{0.1cm} \textbf{for every} found event $h_i$ \textbf{do} \\
 \hspace*{0.8cm} calculate the distance $d(x,h_i)$ \\
 \hspace*{0.8cm} insert the event $h_i$ in the priority queue $\mathcal Q$ according to $d(x,h_i)$ \\
 \hspace*{0.8cm} \textbf{if} $h_i$ is a vertex of $\mathcal K$ \textbf{then}\\
 \hspace*{1.4cm} construct $\gamma(z,h_i)$ \\
 \hspace*{1.4cm} associate to $h_i$ the measurements of the angles formed between $\gamma(z,h_i)$\\
 \hspace*{1.6cm} and the edges incident to $h_i$ inside $F$\\
 \hspace*{0.8cm} \textbf{end} \\
 \hspace*{0.1cm} \textbf{end}
}}}
\end{center}

We will study the running time of the presented algorithm. First we formulate the following lemmas.
\begin{lemma}\label{liniar1}
The shortest path map SPM($x$) of CAT(0) planar complex $\mathcal K$ with $n$ vertices contains $O(n)$ cones.
\end{lemma}

\begin{proof}
By Euler's formula, the number of faces $f$ and the number of edges $m$ of the complex $\mathcal K$ with $n$ vertices are such that $f \leq 2n-4$ and $m \leq 3n-6$.

We affirm that every vertex $z$ of $\mathcal K$ is the apex of a linear number $O(deg (z))$ of cones in SPM($x$).
Indeed, since the angle of any vertex $v$ inside a face $F$ of $\mathcal K$ is strictly less than $\pi$ then $v$ can be the apex of at most one cone inside $F.$
Moreover, since the number of faces of $\mathcal K$ incident to $v$ is $O(deg(z)),$ then $v$ can define a $O(deg(v))$ number of cones of SPM($x$).
Now, using the graph property $\sum deg(h) = 2m$ and $m \leq 3n-6$, we can deduce that the total number of cones of SPM($x$) is of order $O(6n).$
Therefore, SPM($x$) contains $O(n)$ cones.
\hfill $\Box$
\end{proof}

\begin{lemma}\label{liniar2}
Given the shortest path map SPM($x$) inside a CAT(0) planar complex $\mathcal K$ with $n$ vertices, a cone of SPM($x$) can pass via $O(n)$ faces $\mathcal K$.
\end{lemma}

\begin{proof}
Let $\mathcal C(z; p, q)$ be a cone of SPM($x$). A triangular face of $\mathcal K$ is a convex set, and by (ii) of Lemma \ref{pr-spm}, each cone of SPM($x$) is a convex set inside $\mathcal K.$ Thus, $\mathcal C(z; p, q)$ can cross only once an edge of $\mathcal K$.

Since the number of faces of $\mathcal K$ is of order $O(n)$, this implies that a cone $\mathcal C(z; p, q)$ of SPM($x$) can pass via at most $O(n)$ of faces in $\mathcal K.$ \hfill $\Box$
\end{proof}

The following theorem describes the running time of the algorithm and the space used for computing the shortest path map inside a CAT(0) planar complex.

\begin{theorem} \label{compl}
Given a CAT(0) planar complex $\mathcal K$ with $n$ vertices and a point $x$ of $\mathcal K$, it is possible to construct the shortest path map SPM($x$) of $\mathcal K $ in $O(n^2\log n)$ time and $O(n^2)$ space.
\end{theorem}

\begin{proof}
The size of the data structure is defined by the sizes of the two data substructures (static and dynamic).
The static data substructure contains the planar map of the $\mathcal K$ complex. This representation allows us to use a linear-sized data space in the number of vertices of $\mathcal K $ \cite{BeChKrOv}.
Regarding the dynamic data substructure, it retains for any face $F$ of $\mathcal K$, the intersection of $F$ with cones of SPM($ x $).

By Euler's formula, the number of faces $f$ and the number of edges $m$ of $ \mathcal K $ with $n$ vertices, are such that $f \leq 2n-4$ and $ m \leq 3n-6.$
We claim that every vertex $z$ of $\mathcal K$ is the apex of $O(deg (z))$ cones of SPM($ x $).
Indeed, since the angle with the origin in any vertex $v$ inside a face $F$ of $\mathcal K $ is strictly less than $\pi$ then $v$ is the apex of at most one cone inside $F.$

From the proof of the previous lemma, we established that SPM($x$) contains $O(n)$ cones.
Since the number of faces of $\mathcal K$ is of order $O(n)$ we say that the dynamic data substructure uses $O(n^2)$ space.

Let us consider the example shown in Figure \ref{des-quadratic}, for which the size of the data structure is $O(n^2)$. In this complex, $n/2$ vertices are located on three geodesic paths passing via $x$ and the other $n/2$ vertices are located on the boundary of $\mathcal K$ and do not belong to these three geodesics. Therefore, the vertices located on the boundary of $\mathcal K$ contribute to the formation of the $O(n/2)$ cones in SPM($x$). The other $n/2$ vertices contribute to the formation of only three cones. In this case, a cone $\mathcal C(z; p, q)$ of SPM($x$) intersects $O(n/2)$ number of edges of $\mathcal K.$ Adding, we obtain that all the $O(n/2 +3)$ cones of SPM($x$) intersect each $O(n/2)$ edges of $\mathcal K.$ This shows that the size of the data structure is quadratic with respect to the number of vertices of $\mathcal K.$

\begin{figure}[h]\centering
\includegraphics[width=0.39\textwidth]{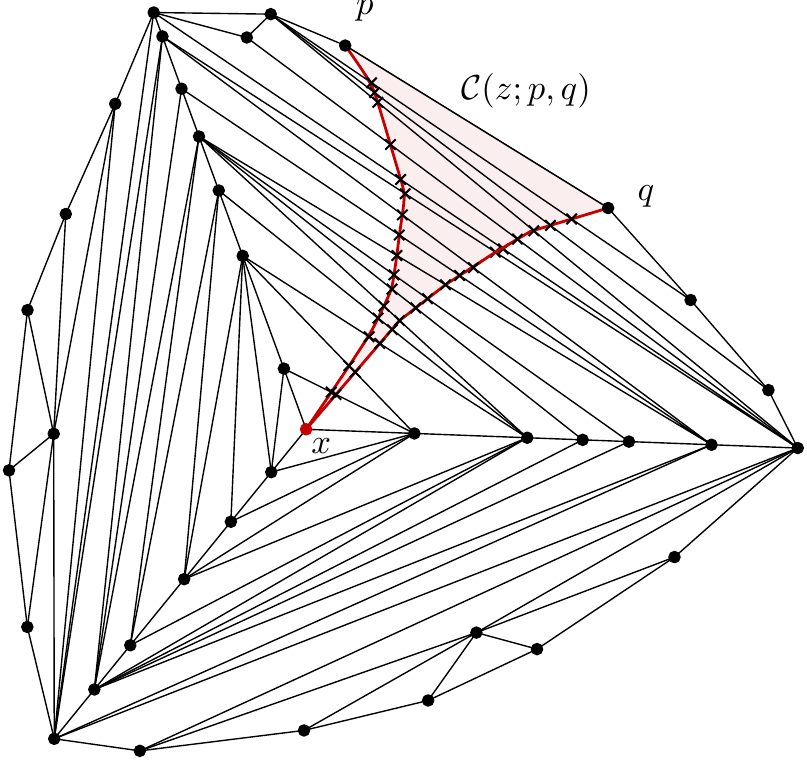}
\caption{A cone of SPM($x$) cutting $O(n)$ edges of $\mathcal K$.}\label{des-quadratic}
\end{figure}

In order to estimate the running time of the algorithm, we calculate the execution time of each step of the algorithm.
We recall that at each step of the algorithm, an event is extracted from the priority queue $\mathcal Q$ to be processed according to its type. The algorithm terminates when $Q$ is empty. Note that an event is crossed by the sweeping line only once. Therefore, the number of steps of the algorithm is equal to the number of sweep events.
By definition of a vertex-event, sweeping line crosses $n$ such events. Regarding edge-events, by Lemma \ref{liniar1}, an edge can be crossed by $O(n)$ cones, each cone creating an edge-events. By Euler's formula, the complex contains at most $3n-6$ edges. The sweeping line crosses therefore $O(n^2)$ edge-events.

Next we establish the running time of each step of the algorithm. For this, we analyze separately the processing of a vertex-event and an edge-event. For each event of the sweep in order to maintain the event order in the priority queue the algorithm takes $O(\log n)$ time.
Processing of a vertex-event $h$ consists of transmitting the information of the face covered by the sweeping line to the next adjacent face (where $h$ belongs to the common edge) which is performed in constant time, and finding new events and include them in the priority queue which is performed in $O(\log n)$ time.\\
The processing of the vertex-event $h$ consists in building at most $deg(h)$ new cones with origin at $h,$ and passing the information for sweeping from the previous face to the face incident to $h$ and finding new events inside this face. Therefore, the running time for processing a vertex-event is $O(n\log n)$.

In summary, the algorithm performs $O(n^2)$ processing steps of edge-events, performed in $O(\log n)$ time and $O(n)$ processing steps of vertex-events, where each step is performed in $O(n\log n)$ time. \\
The running time of the algorithm is thus $O(n^2\log n)$ with respect to the number of vertices of $\mathcal K$. \hfill $\Box$
\end{proof}

\section{One-point shortest path queries}

In this section, we present the detailed description of the algorithm for answering one-point shortest path queries in CAT(0) planar complexes and of the data structure $\mathcal D$ used in this algorithm. The algorithm is linear with respect to the number of vertices of the complex and uses the same data structure as in computing the SPM($x$).
Given a CAT(0) planar complex $\mathcal K$, a point $x\in \mathcal K$ and the shortest path map SPM($x$), for every query point $y\in\mathcal K$ first we determine the face containing $y$. Let $F$ be the face containing $y$. Using the data structure, we locate $y$ in a cone $\mathcal C(z; p,q)$ of SPM($x$) which crosses $F$. We show then how to embed isometrically $\mathcal C(z;p,q)$ in the plane as an acute triangle. Let $f$ be the isometric unfolding of the cone in $\mathbb R^2.$ Then the preimage of the shortest path $\gamma(f(x),f(y))$ between the images of $x$ and $y$, is exactly the shortest path between $x$ and $y$ in $\mathcal K$ (see Fig. \ref{des_pcc}).

\begin{figure}[h]\centering
\includegraphics[width=0.55\textwidth]{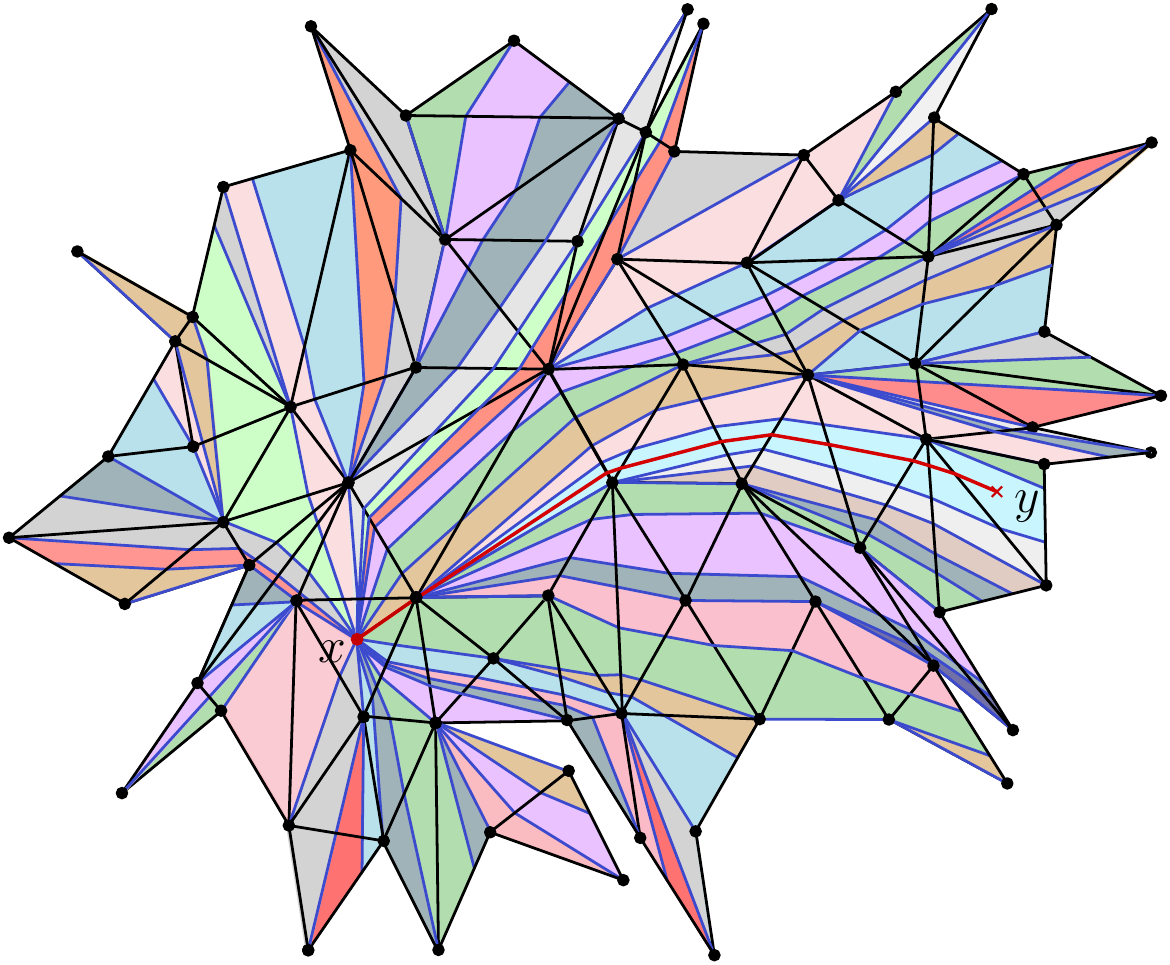}
\caption{Shortest path between $x$ and a point $y$ of $\mathcal K$.} \label{des_pcc}
\end{figure}

\subsection{Computing the shortest path}

Let $y$ be an arbitrary point of $\mathcal K.$ In order to construct the shortest path $\gamma(x, y),$ first, we determine the cone $\mathcal C (z; p^*, q^*)$ of SPM($x$) containing $y$. By Lemma \ref{plonj-con}, this cone can be unfolded in the Euclidean plane in the form of an acute triangle denoted by $\mathcal T$. Our algorithm constructs the shortest path $\gamma(y, z)$ inside $\mathcal K$ as the preimage of the shortest path $\gamma(f(y), f(z))$ inside $\mathcal T,$ where $f$ is the unfolding of $\mathcal C (z; p^*, q^*) $ in the plane. \\
We describe later the steps of the algorithm in a more detailed way.

\subsubsection{Locating $y$ inside a cone of SPM($x$)}

Suppose first that $y$ is an inner point of a triangular face $F = \Delta(a, b, c)$ of $\mathcal K.$ Using the data structure described in the previous section and the local coordinates of the point $y$ inside $F,$ the algorithm determines by a binary search the cone of SPM($x$) which intersects $F$ and contains $y$. We reconstruct the quadrilateral using the given coordinates of the intersection points of $F$ with all cones of SPM($x$). Since all cones of SPM($x$) are convex sets, they intersect the edges of a face of $\mathcal K$ only once. Therefore, we have two possible cases: either the sides of cones ''enter'' via a single edge of $F$ and ''go out'' via the other two edges of $F$, 
or the the sides of cones ''enter'' via two edges of $F$, and ''go out'' via a single edge of $F$. 

Thus $F$ is divided into a finite number of quadrilaterals denoted by $Q_i$, $i\in I\subset \mathbb{N}$.
Once the sides of quadrilaterals inside $F$ are constructed, the algorithms locates $y$ in one of them using a binary search among the sides of $Q_i$. The quadrilateral $Q^*$ containing $y$ is defined as the quadrilateral for whom $y$ is located at left with respect to one side and at right to respect to the other side of the quadrilateral.
Using the static data structure $D_s$ we can identify the cone $\mathcal C(z; p^*,q^*)$ of SPM($x$) containing the quadrilateral $Q^*$.

\subsubsection{Reconstruction and unfolding of cones}


Let $\mathcal C(z; p^*,q^*)$ and $F = \Delta(a, b, c)$ be respectively the cone of SPM($x$) and the triangular face of $\mathcal K$ containing $y$. Using the static data structure $D_s$, we can efficiently retrieve the segments forming the sides of the cone from the lists of intersections points of edges in $\mathcal K$ with the sides of the cone $\mathcal C(z; p^*,q^*)$. In the same way, we calculate the angle inside the cone whose apex is $z$.

\begin{figure}[h]\centering
\includegraphics[width=0.8\textwidth]{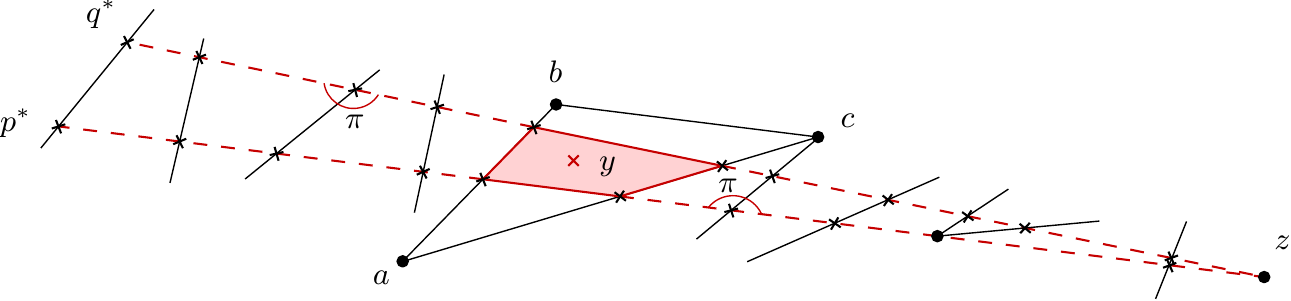}
\caption{Reconstruction of the cone in the plane.}\label{recon-con}
\end{figure}

Let $f$ be the unfolding of $\mathcal C(z; p^*,q^*)$ in the plane.
Once the sweeping line passed through the sequence of adjacent faces $F_1, F_2, \ldots, F_k$ of $\mathcal K$ starting with $F,$ it is possible to recover all the segments forming the sides of the cone $\mathcal C(z; p^*,q^*)$.
Then we build the images of these segments in the plane, starting with the point $f(z)$ and joining them in such a way that the obtained complementary angles between two connected segments are equal to $\pi$ (see Fig. \ref{recon-con}). Moreover, we build segments belonging to the edges of $\mathcal K$ contained in the cone $\mathcal C(z;p^*,q^*)$.

\begin{figure}[h]\centering
\includegraphics[width=0.85\textwidth]{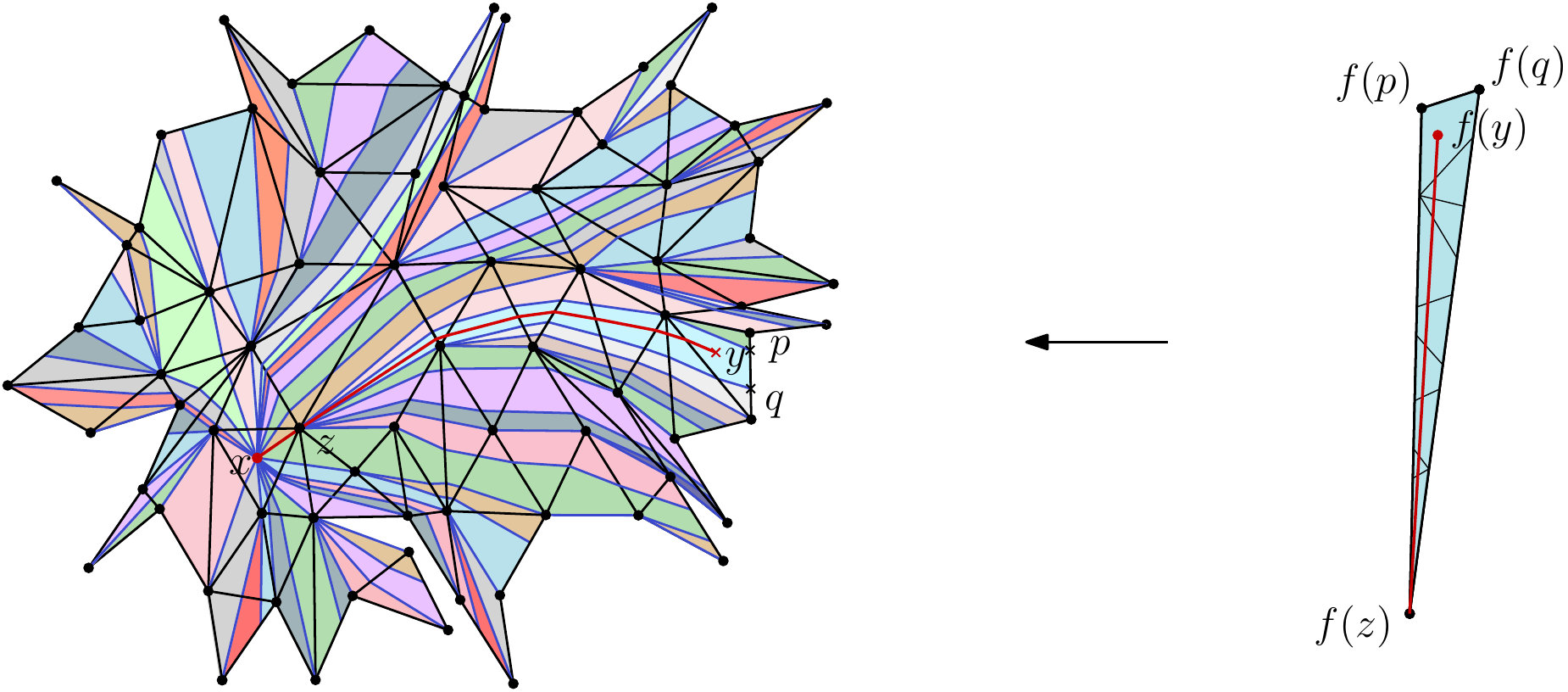}
\caption{Embedding a cone in the plane.}\label{des-depliage}
\end{figure}
The images of the segments forming the cone $\mathcal C(z; p^*,q^*)$ are pairwise concatenated so that the formed angles on one side and the other are equal to $\pi$. Moreover, the images of the segments forming the angle of origin $z$ inside $\mathcal C(z; p^*,q^*)$ which is less than or equal to $\pi,$ form an equal angle in the plane. Therefore, we obtain the unfolding of $\mathcal C(z; p^*,q^*)$ in $\mathbb {E}^2$ in the form of an acute triangle (see Fig. \ref{des-depliage}). We denote by $\mathcal T = \Delta(f(z), f(p^*),f(q^*))$ the triangle obtained in the plane.

\subsubsection{Computing the shortest path}

At this stage, we locate the image $f(y)$ and $y$ inside $\mathcal T$ using the coordinates of the images of the intersection points of edges in $F$ with the segments belonging to sides of $\mathcal C(z; p^*,q^*).$ Afterwards we construct the geodesic $\gamma(f(z),f(y))$ inside $\mathcal T$ using the Euclidean metric.

In order to construct the shortest path $\gamma(z, y)$ inside $\mathcal K$, we use the pre-images of the intersection points between the constructed geodesic $\gamma(f(z), f(y))$ with the images of segments belonging to edges of $\mathcal K.$
Therefore, the geodesic $\gamma(y, z)$ is constructed inside $\mathcal K$ by concatenating segments built between the pre-images of these intersection points.

Since $z$ is a vertex of $\mathcal K$, the geodesic $\gamma(x, y)$ is constructed when computing SPM($x$).
By Lemma \ref{sp}, the geodesic $\gamma(x, y)$ is the concatenation of geodesics $\gamma(x, z)$ and $\gamma(z, y)$.

\subsection{Algorithm}

Summarizing the results of the previous subsections, we present the algorithm for calculating the distance and the shortest path between a given point and any other point in a CAT(0) planar complex.

\bigskip
\begin{center}
\framebox{
\parbox{15cm}{
\vspace{0.05cm}
\noindent{\bf Algorithm} {\sc Computing the shortest path for one-point distance queries problem}\\
{\footnotesize
  {\bf Input:} The shortest path map SPM($x$) of a CAT(0) planar complexe $\mathcal K$, a point $y,$ a data structure \\
  {\bf Output:} The shortest path $\gamma(x,y)$ between $x$ and $y$ inside $\mathcal K$ \\
\rule{\linewidth}{.7pt}
 \hspace*{0.1cm} \textbf{if} $y$ is a vertex inside $\mathcal K$ \textbf{then}\\
 \hspace*{0.8cm} return $\gamma(x,y)$ contained inside SPM($x$) \\
 \hspace*{0.1cm} \textbf{else} \\
 \hspace*{0.8cm} find the cone $\mathcal C(z;p^*,q^*)$ of SPM($x$) containing $y$ \\
 \hspace*{0.8cm} construct the unfolding $f$ of $\mathcal C(z;p^*,q^*)$ in the plane. Let $\mathcal T=f(\mathcal C(z;p^*,q^*))$ \\
 \hspace*{0.8cm} locate $f(y)$ inside $\mathcal T$ \\ 
 \hspace*{0.8cm} computing $\gamma(f(y),f(z))$ as the shortest path between $f(y)$ and $f(z)$ inside $\mathcal T$\\
 \hspace*{0.8cm} return $f^{-1}(\gamma(f(y),f(z)))\cup\gamma(z,x)$ \\
 \hspace*{0.1cm} \textbf{end}
}}}
\hspace{0.2cm}
\end{center}
\medskip

\begin{lemma}\label{compl-depliage}
The unfolding of a cone of SPM($x$) is computed in $O(n)$ time.
\end{lemma}

\begin{proof}
Let $\mathcal C(z; p, q)$ be a cone of SPM($x$). By Lemma \ref{liniar2}, $\mathcal C(z; p, q)$ intersects $O(n)$ faces of $\mathcal K.$ All cones of SPM($x$) are convex and so we can deduce that $\mathcal C(z; p, q)$ is cutting a face of $\mathcal K$ exactly once. Therefore, a cone of SPM($x$) is divided into $O(n)$ quadrilaterals. Note that the intersection of cone with a face may also contain triangles, in this case, we consider these triangles as degenerate quadrilaterals.

Let $F$ be a face intersected by the cone $\mathcal C(z; p, q)$ and let $Q = abcd$ be a quadrilateral obtained as the intersection of $F$ with $\mathcal C(z; p, q)$. The local coordinates of the points $a, b, c$ and $d$ inside $F$ are known, thus, we can construct the isometric image of $Q$ in the plane in a constant time. Therefore, if the cone cuts $O(n)$ faces of $\mathcal K$ then the total time of the unfolding of $\mathcal C(z; p, q)$ is $O(n)$. This embedding of $\mathcal C(z; p, q)$ is made by joining pairwise adjacent quadrilaterals (which share a common side).
\hfill $\Box$
\end{proof}

\noindent
\begin{theorem} \label{complexit}
Given a CAT(0) planar complex $\mathcal K$ with $n$ vertices, one can construct a data structure of size $O(n^2)$ and in $O(n^2\log n)$ time such that, for a query $x \in \mathcal K,$ the algorithm computes the shortest path between $x$ and any other point $y$ of $\mathcal K$ in $O(n)$ time.
\end{theorem}

\begin{proof}
The algorithm uses the data structure from Theorem \ref{compl} and of size $O(n^2)$.
We show below that the algorithm runs in linear time. First, we analyze each step of the algorithm. We begin with locating the point $y$ in a cone of SPM($x$).
Let $F$ be the face of the complex containing the query point $y$. We use the data structure and the local coordinates of $y$ inside $F$, to guide a binary search in $F$ which determines the quadrilateral $Q$ of $F$ containing the point $y.$ The required time for locating $y$ in a quadrilateral of $F$ is $O(\log n)$. \\
Subsequently, the algorithm reconstructs the cone $\mathcal C (z; p, q)$ of SPM($x$) containing $Q$ inside the sequence of pairwise adjacent faces starting from $F$ in both directions (in one direction up to the boundary of $\mathcal K$ and up to the face containing $z$ in the other direction). By Lemma \ref{liniar2} the cone $\mathcal C(z; p, q)$ intersects a linear number of faces of $\mathcal K.$ Therefore, the reconstruction of the cone containing $y$ requires linear time with respect to $n$. \\
Thus, the total running time of this step is $O(n)$.

We analyze the next step which is the unfolding in the plane of the cone of SPM($x$) containing $y$. Let $f$ be the isometric embedding of $\mathcal C(z; p, q)$ in $\mathbb{E}^2$. By the Lemma \ref{compl-depliage}, the required time to complete this step is $O(n)$.

Finally, the last step of the algorithm is to compute the shortest path $\gamma(x, y)$ inside $\mathcal K.$
Let $\mathcal T$ be the unfolding of $\mathcal C(z; p, q)$ in the plane. The algorithm constructs the shortest path between images $f(x)$ and $f(y)$ inside the Euclidean triangle $\mathcal T$ in a constant time. Then, the algorithm builds inside the quadrilaterals forming $\mathcal C(z; p, q)$, the pre-images of the segments of $\mathcal T$ form $\gamma(f(x), f(y))$. This operation requires a constant time. Thus, the shortest path $\gamma(x, y)$ can be computed inside $\mathcal K$ in $O(n)$ time.

Therefore, by summing the requested time of each step of the algorithm, the algorithm computes the shortest path between a given point $x$ and any other point $y$ inside a CAT(0) planar complex $\mathcal K,$ in total $O(n)$ time. \hfill $\Box$
\end{proof}

\section{Convex Hull}

\noindent In this section we consider the following problem.\\

\textbf{HC$(S)$}: \emph{Given a finite set of points $S$ in a CAT(0) planar complex $\mathcal K$, construct the convex hull of $S$ inside $\mathcal K$ (see Fig. \ref{des-hc}).}\\

\noindent In order to solve this problem, we present an algorithm based on the same principle as Toussaint's algorithm \cite{To}, which constructs the convex hull of a finite set of points in a simple triangulated polygon in $O(n \log n)$ time, using a data structure of size $O(n).$
As a preprocessing step, our algorithm builds the shortest path map SPM($x$), where $x$ is any point of the boundary of the complex. Thus, all cones of SPM($x$) are sorted starting with the cone having one side on the boundary of $\mathcal K.$
\begin{figure}[h]\centering
\includegraphics[width=0.4\textwidth]{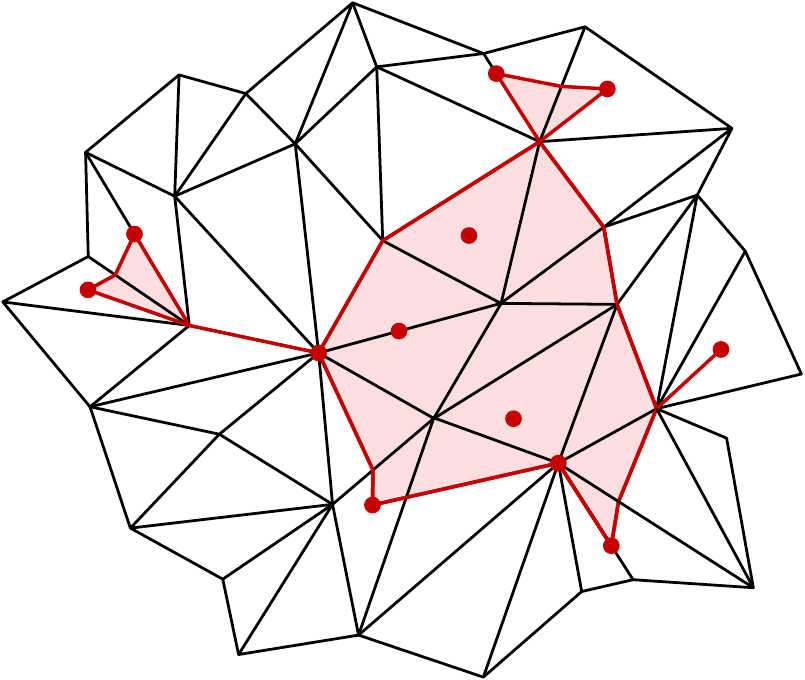}
\caption{The convex hull of a finite set of points in a CAT(0) planar complex.}\label{des-hc}
\end{figure}
We note that this partition in cones plays the role of a triangulation of a simple polygon. We build in each cone the convex hull of the subset of $S$. By connecting every pair of the computed convex hulls in the same order as the cones that contain them, we obtain a weakly-simple polygon inside $\mathcal K.$ We denote by $P$ the union of the obtained weakly-simple polygon with a geodesic segment $\gamma(p, q),$ where $p$ and $q$ are the closest points, such that $p$ belongs to one of the constructed convex hulls and $q$ belongs to the boundary of $\mathcal K$. We show that the complex $\mathcal K \setminus P$ is a CAT(0) planar complex. By considering the point $p^*$ as the copy of $p$ such that $p$ is different from $p^*$ inside $ \mathcal K \setminus P,$ we show that the shortest path between $p$ and $p^*$ inside the CAT(0) planar complex $\mathcal K \setminus P$ is exactly the boundary of the convex hull of $S$ inside $\mathcal K.$

\subsection{Algorithm}

The data structure $D$ contains the local coordinates of all the points of the set $S$ inside the faces of $\mathcal K$ and an ordered list of cones $\mathcal L_C$ of SPM($x$). The points of $S$ can be any points of the complex, including the points on the boundary of $\mathcal K.$\\
Let us recall first some important concepts: a closed polygonal line of the complex is called a \emph{weakly-simple polygon} if the polygonal line divides the complex into two areas equivalent to a disc. Less formally, we can say that a weakly-simple polygon can have sides that touch but do not cross.

Let $k$ be the number of points of the set $S$. We denote by $S_1, S_2, \ldots, S_l$ the subsets of $S$ such that two sets $S_j$ and $S_{j+1}$ with $1\geq j \leq l-1,$ belong to two distinct cones of $\mathcal L_C$.

We denote by $\mathcal L_C(S)$ the sublist of the list of cones $\mathcal L_C,$ in which every cone contains at least one point of $S$. Thus, if $\mathcal C(z;p,q)\subset \mathcal L_C(S)$, then $\mathcal C(z;p,q)\cap \neq\emptyset.$
Further, inside each of the cones of $\mathcal L_C(S)$ we construct the convex hull of the points of $S$ as follows. We recall that for any cone $\mathcal C(z; p, q)$ of SPM($x$), int$(\mathcal C (z; p, q))$ does not contain vertices of $ \mathcal K.$
It is possible to determine the ordered sequence of faces cut by a cone of SPM($x$). Thus, for a cone $\mathcal C(z; p, q)$ of $\mathcal L_C(S)$, we determine the largest ordered subsequence of consecutive faces $F_i,F_{i+1},\ldots,F_k$ so that the extremal faces $F_i, F_k$ contain points of $S$. By building in the plane the isometric images of the faces $F_i,F_{i+1},\ldots,F_k$, as well as the points of $S_i \subseteq S$ insides these faces, we then construct the convex hull of the images of these points in $\mathbb{E}^2.$ The pre-image of the convex hull computed in the Euclidean plane is exactly the convex hull inside $\mathcal K$ of all points of $S_i,$ such that $S_i \cap \mathcal C(z; p, q) \neq \emptyset$.

Since $x$ is chosen on the boundary of $\mathcal K,$ SPM($x$) can be seen as a canonical counterclockwise non-cyclic sequence of cones. By setting an ordering of cones of SPM($x$), i.e. the cones of $\mathcal L_C$, we automatically obtain an ordering of the cones in the sublist $\mathcal L_C(S).$
It is possible to choose the exact ordering of the cones of $\mathcal L_C(S)$ as of the subsets $S_1,S_2,\ldots,S_l$ of $S$.
Thus, for each pair of consecutive cones of $\mathcal L_C(S)$, containing consecutive subsets $S_{i-1}, S_i \subset S$ respectively, using the shortest path algorithm presented previously, we construct a geodesic segment $\gamma(a_{i-1 },b_i)$ between an arbitrary point $a_{i-1}$ of $S_{i-1}$ and an arbitrary point $b_i$ of $S_i.$ This operation allows us to connect the convex hulls conv$(S_i),$ $i = 1, \ldots, l$ to obtain a weakly-simple polygon.

Subsequently, we want to connect the weakly-simple polygon obtained inside $\mathcal K$ with boundary of $\mathcal K.$ First, we choose two points $p\in S$ and $q \in \partial \mathcal K$ and construct the geodesic segment $\gamma(p, q)$.
The points $p$ and $ q $ can not be chosen arbitrarily. It is necessary that $p$ is on the boundary of the convex hull of $S$, thus, $p \in P$ is the nearest point of $S$ to the boundary of $\mathcal K.$ Therefore, we choose $p$ such that $d(p,q)=\displaystyle\min_{\substack{t\in S_1\\ t'\in \partial \mathcal K}} d(t,t').$ This calculation is possible by embedding in the plane the extreme cone of $\mathcal K$ containing $S_1$. \\
The choice of $p$ and $q$ as described above will insure that $p$ is on the boundary of the convex hull conv$(S)$ inside $\mathcal K$ and that $p$ and $q$ belong to a single cone of $\mathcal L_C.$

We construct the set $P:=(\displaystyle\cup_{i=1}^{l}$conv$(S_i))\cup(\displaystyle\cup_{i=2}^l \gamma(a_{i-1},b_i))\cup \gamma(p,q).$ Note that $P$ is defined as the union of weakly-simple polygon constructed above and the geodesic $\gamma(p, q).$ The constructed set $P$ is an weakly-simple polygon. We duplicate the points $p$ and $q$. Let $p^* = p$ and $q^* = q$ be pairwise distinct points in $\mathcal K \setminus P.$
In the complex CAT(0) planar complex $\mathcal K \setminus P$ we construct the shortest path $\gamma(p, p^*)$ using the algorithm of the previous section. The boundary of the convex hull of $S$ inside $\mathcal K$ coincides with the shortest path $\gamma(p, p^*)$ inside $\mathcal K \setminus P.$

\noindent Here is a brief and informal description of the presented algorithm.

\bigskip
\begin{center}
\framebox{
\parbox{15cm}{
\vspace{0.05cm}
\noindent{\bf Algorithm} {\sc The Convex hull of a finite set of points}\\
{\footnotesize
  {\bf Input:} a CAT(0) planar complex $\mathcal K$, SPM($x$) with $x\in\partial\mathcal K$, a finite set of points $S\subset\mathcal K$ and data structure $D$\\
  {\bf Output:} the convex hull conv$(S)$ given by its boundary \\

\hspace*{0.1cm} determine the subsets $S_i,$ $i=1,\ldots,l$ where $l\leq |S|$ \\
\hspace*{0.1cm} \textbf{for} every $i\leftarrow 1$ to $l$ \\
\hspace*{0.6cm} compute conv$(S_i)$ \\
\hspace*{0.1cm} \textbf{end for} \\
\hspace*{0.1cm} \textbf{for} every $i\leftarrow 1$ to $l-1$ \\
\hspace*{0.6cm} choose $a_i\in S_i$, $b_{i+1}\in S_{i+1}$ and compute $\gamma(a_i,b_{i+1})$ \\
\hspace*{0.1cm} \textbf{end for} \\
\hspace*{0.1cm} determine $p\in S$ and $q\in \partial\mathcal K$, such that $d(p,q)=\displaystyle\min_{\substack{t\in S_1\\ t'\in \partial \mathcal K}} d(t,t')$ \\
\hspace*{0.1cm} compute $P:=(\displaystyle\cup_{i=1}^{l}$conv$(S_i))\cup(\displaystyle\cup_{i=2}^l \gamma(a_{i-1},b_i)) \cup \gamma(p,q)$ \\
\hspace*{0.1cm} double the points $p=p^*$ and $q=q^*$ \\
\hspace*{0.1cm} compute $\gamma(p,p^*)$ inside $\mathcal K\setminus P$ \\
\hspace*{0.1cm} return $\partial$conv$(S) = \gamma(p,p^*)$ \\
\vspace{0.4cm}
}}}
\end{center}

\noindent We prove now the correctness of the algorithm.
\begin{proposition}\label{KminusK0}
The complex $\mathcal K \setminus P$ is a CAT(0) planar complex.
\end{proposition}

\begin{proof}
We will use a result of Gromov stating that a complex is CAT(0) if and only if it is simply connected and has a nonpositive curvature in each point of the complex.
First we prove that $\mathcal K \setminus P$ is simply connected.
A given space is simply connected if every loop drawn in this space can be continuously reduced (by homotopy) to a point.
$P$ is a union of geodesic segments and convex sets and therefore it is a weakly-simple polygon $\mathcal K.$
Let $\mathcal K_0$ be the union of a geodesic segment $\gamma(x_1, y)$ and a closed polygonal line $x_1, x_2, \ldots, x_k, x_1$, such that $x_i, x_{i+1}$ belong to one face of $ \mathcal K$. We will show that $\mathcal K \setminus \mathcal K_0$ is a CAT(0) planar complex. By induction on the number of figures equivalent to $\mathcal K_0$, we obtain that $\mathcal K \setminus P$ is CAT(0).

Let $L = e_1, e_2, \ldots, e_s$ be a closed chain of $\mathcal K,$ where $e_i \in E(\mathcal K)$, for all $i = 1, \ldots, s$.
We can distinguish two cases: the case where $L \cap \mathcal K_0 = \emptyset$ and the case where $L \cap \mathcal K_0 \neq \emptyset$.

\begin{figure}[h]\centering
\begin{tabular}{cc}
   \includegraphics[width=0.4\textwidth]{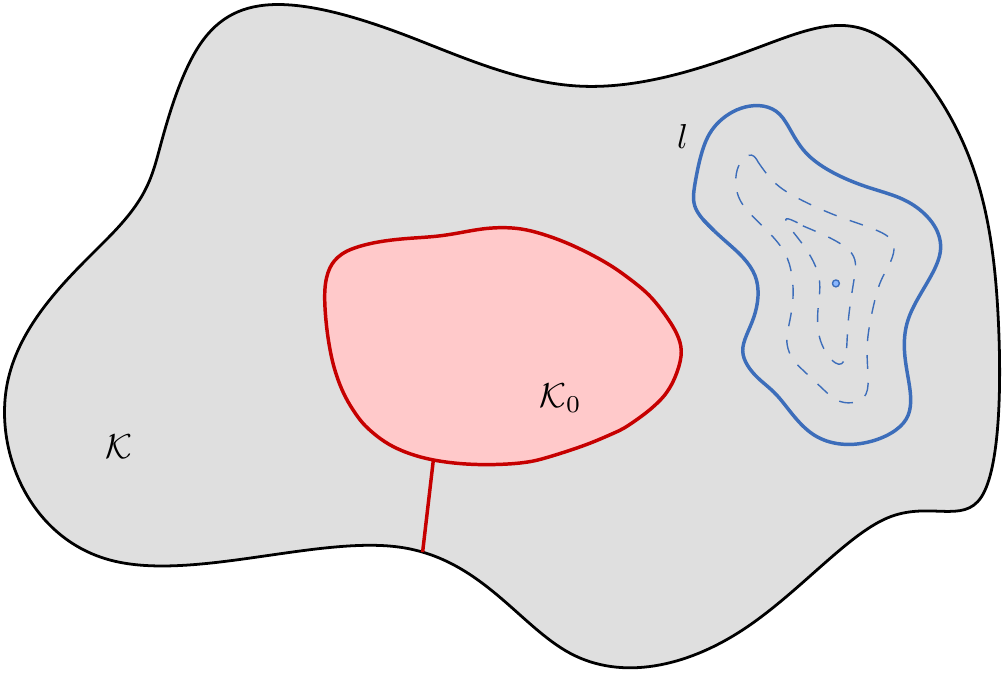}
   &
   \includegraphics[width=0.4\textwidth]{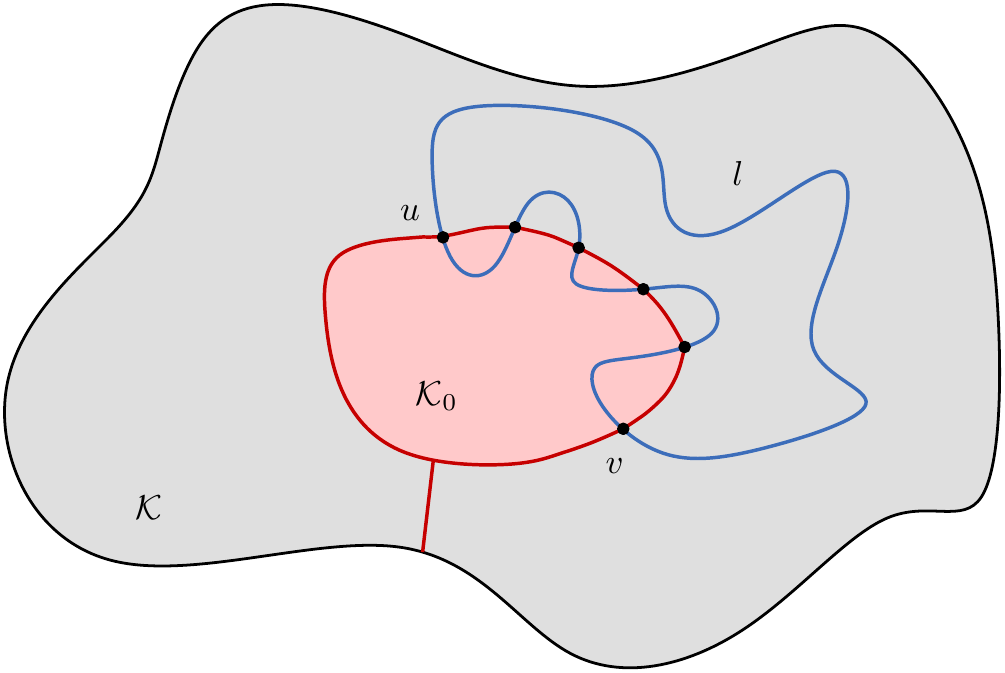}
    \\
   (a)
   &
   (b)
\end{tabular}
 \caption{For the Proposition \ref{KminusK0}, (a) $L\cap \mathcal K_0 =\emptyset,$ (b) $L\cap \mathcal K_0\neq \emptyset$.} \label{deslatz}
\end{figure}

\noindent First we analyze the case where $L \cap \mathcal K_0 = \emptyset$ (see Fig. \ref{deslatz} (a)). In this case, since the chain $L$ does not intersect $\mathcal K_0,$ then $L$ is contractible in $\mathcal K \setminus \mathcal K_0$.

\noindent In the case where $L \cap \mathcal K_0 \neq \emptyset$ (see Fig. \ref{deslatz} (b)), there exist at least two vertices $u$ and $v$ such that $L \cap \mathcal K_0 = \{u, v \}$ and $d_G (u, v)$ is maximal. These points $u$ and $v$ belong to $L.$ On the other hand, since $u, v \not \in \mathcal K \setminus \mathcal K_0$ then $L$ is no longer a closed chain in $\mathcal K \setminus \mathcal K_0.$ \\
Therefore, in both cases, $\mathcal K \setminus \mathcal K_0$ is simply connected.

Let $\mathcal K_i$ be the union of $t$ figures of type $\mathcal K_0$ 
pairwise connected by a point (see Fig. \ref{desconv(P)}). By induction on the number of figures $\mathcal K_i \subset \mathcal K$, the subcomplex $\mathcal K \setminus \displaystyle \cap_{i = 1}^t \mathcal K_i$ is simply connected.


It remains to be shown that for every point $x \in \mathcal K \setminus P $ the curvature in the point $x$ is nonpositive.
Let us consider $x$ inside int$(\mathcal K \setminus P).$ Since $\mathcal K \setminus P$ is a subcomplex of $\mathcal K$ then $x$ belongs to int$(\mathcal K)$ and therefore the curvature is nonpositive in the point $x$.

\noindent For any point $x$ on the boundary of $\mathcal K \setminus P$ we can distinguish two cases: $x$ belongs to the boundary of $\mathcal K$ and $x$ belongs to the boundary of $P.$
The first case is trivial, thus let us consider the second case. Suppose $x$ belongs to $\partial P.$
The curvature of the initial complex $\mathcal K$ is nonpositive, i.e., $\theta(x) \geq 2\pi.$ The curvature in the point $x$ of the subcomplex $\mathcal K \setminus P$ the sum the angles of the faces incident to $x$ in this point can be only less than $\theta(x)$ inside $\mathcal K.$ Therefore, the curvature of the subcomplex $\mathcal K \setminus P$ is non-positive in the point $x$.

The complex $\mathcal K \setminus P$ is a CAT(0) planar complex. We will prove that it is of the same type as the complex $\mathcal K$.
The complex $\mathcal K \setminus P$ contains faces of length greater than $3,$ since after eliminating $\mathcal K(P)$ of the initial complex $\mathcal K$, we have removed parts of faces of $\mathcal K.$

It is always possible to transform the subcomplex $\mathcal K \setminus P$ into a complex of the same type as the initial complex $\mathcal K$ by constructing diagonals inside the obtained faces which divides them into smaller faces of length equal to 3.
\hfill $\Box$
\end{proof}

We recall that the point $p$ is chosen so that it belongs to the boundary of the convex hull of $S$ and the point $p^*$ is the copy of $p$, then the following lemma is satisfied.

\begin{proposition}\label{gamma_p-p}
The constructed geodesic $\gamma(p, p^*)$ is the boundary of the convex hull of $S.$
\end{proposition}

\begin{proof}
Suppose that $\gamma(p, p^*)$ does not coincide with the boundary of conv$(P)$.
We say that $\partial$conv$(P)$ is a subset of $\mathcal K \setminus P.$ Otherwise, there exists a point $x \in \partial$conv$(P)$ such that $x \not \in \mathcal K \setminus P,$ and therefore, $x$ belongs to int$(P)$ which contradicts the fact that $x$ is an extreme point of conv$(P)$.

Since $P$ is a weakly-simple polygon, the boundary of conv$(P)$ is a chain of alternating geodesic segments in $\mathcal K \setminus P$ and of geodesic segments of $\partial P.$
We prove that $\partial$conv$(P)$ is convex inside the complex $\mathcal K \setminus P.$

Suppose the contrary. Let $x$ and $y$ be two points on the boundary of conv$(P)$. We consider the geodesic segment $\gamma(x, y)$ inside the complex $\mathcal K \setminus P$ and there exists at least one point $z \in \gamma(x, y),$ such that $z \not \in \partial$conv$(P)$. The point $z$ belongs to the interior of conv$(P) \setminus P$ which is a subset of $\mathcal K \setminus P$ (see Fig. \ref{desconv(P)}). Since $\partial$conv$(P)$ is a concatenation of geodesic segments inside $\mathcal K \setminus P$ and segments belonging to $\partial P$, there exists at least two points of $\gamma(x, y)$ connected by two distinct geodesics: one belonging to $\gamma(x, y)$ and the other containing $z$. This is in contradiction with the uniqueness of the geodesic connecting two points in a CAT(0) space. Therefore, $\partial$conv$(P)$ is convex in $\mathcal K \setminus P.$

\begin{figure}[h]\centering
\includegraphics[width=0.4\textwidth]{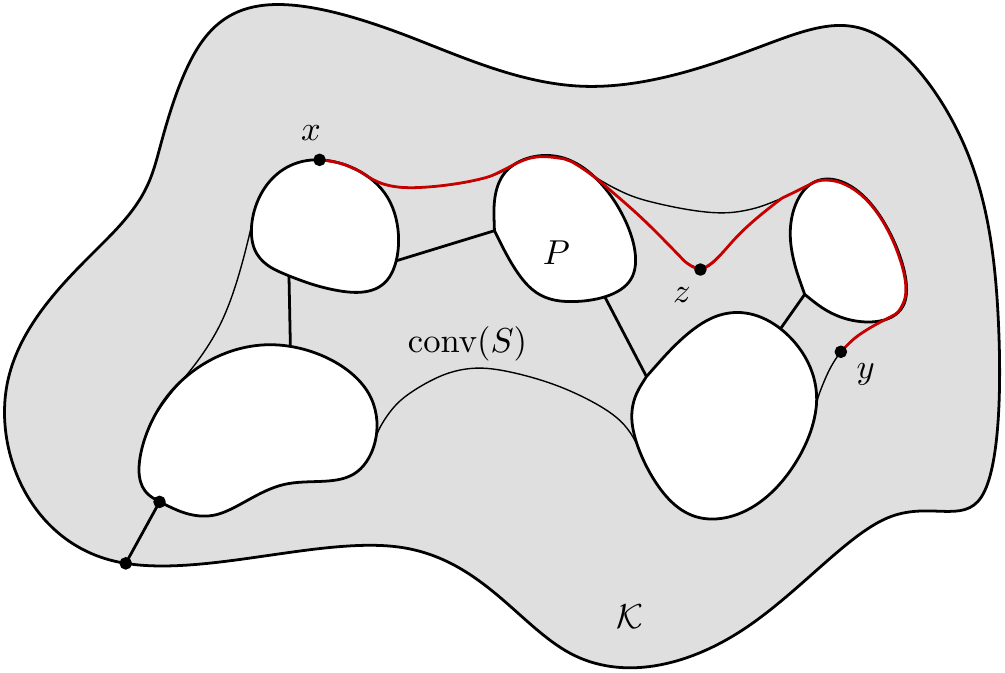}
\caption{The case where the geodesic $\gamma(p,p^*)$ does not coincide with the boundary of conv$(P)$ inside $\mathcal K$.}\label{desconv(P)}
\end{figure}

We have obtained two distinct convex chains $\gamma(p, p^*)$ and $\partial$conv$(P)$ inside $\mathcal K \setminus P$ with the common point $p = p^*$. Since $\gamma(p, p^*)$ is a geodesic, then the length of $\gamma(p, p^*)$ is strictly less than the length of $\partial$conv$(P)$.
Therefore, $\gamma(p, p^*) \subset$conv$(P).$

We will show that the set bounded by the geodesic $\gamma (p, p^*)$ inside $\mathcal K$ is a convex set containing the weakly-simple polygon $P$.
Suppose the contrary and denote by $Q$ the set of points $\mathcal K$ bounded by the geodesic $\gamma(p, p^*)$. Let $x$ and $y$ be two points of int$ (Q)$, such that there exists a point $z\in \gamma(x, y)$ and $z \not\in Q$. Therefore, $\gamma(x, y)$ intersects $\gamma(p, p^*)$ at least twice. We obtain two distinct geodesics connecting the points of intersection of $\gamma(x, y)$ and $\gamma(p, p^*)$ which is impossible. Thus, we have shown that $Q$ is bounded by $\gamma(p , p^*)$ is a convex set. By the construction of $\gamma(p, p^*)$ inside $\mathcal K \setminus P,$ the polygon $P$ is contained in the set $Q$.

Since $Q \subset$conv$(P),$ we obtain a contradiction with the fact that conv$(P)$ is the smallest convex set containing $P$. \hfill $\Box$
\end{proof}

\noindent Let us study the running time of the algorithm and the used space.
\begin{theorem} \label{compl_point-convexe}
Given a CAT(0) planar complex $\mathcal K$ with $n$ vertices, a shortest path map SPM($x$) with $x\in \partial\mathcal K,$ one can construct a data structure of size $O(n^2 + k)$, such that for any finite set $S$ of $k$ points in $\mathcal K$, is it possible to construct the convex hull conv$(S)$ in $O(n^2\log n + nk \log k)$ time.
\end{theorem}

\begin{proof}
We start by describing the amount of the memory space used by our algorithm. Since $x$ is a point on the boundary of $\mathcal K$, SPM($x$) can be seen as a noncyclic sequence of cones and thus we can fix an order of cones in SPM($x$). The data structure stores the coordinates of $k$ points of the set $S$ inside the faces that contain them, and the ordered list of cones of SPM($x$).
We consider the complex $\mathcal K$ to be the complex associated to the smallest ordered sublist of cones of SPM($x$), such as the extreme cones of this sublist contain points of $S$. By the Theorem \ref{compl}, the size of the data structure used in order to construct the shortest path map is of $O(n^2)$ size. \\
Thus, since the implemented data structure contains the coordinates of $k$ points of $S$ inside the faces of $\mathcal K$ and the data substructure of SPM($x$), this requires a memory space of size $O(n^2 + k)$.
At the last step of the algorithm, in order to construct the shortest path between $p$ and its copy $p^*$ inside the subcomplex $\mathcal K \setminus P$, we use the algorithm presented in the previous section. This algorithm uses memory of quadratic size with respect to $n$. \\
In total, the size of the data structure used by the algorithm for constructing the convex hull of $S$ is of order $O(n^2 + k)$.

We show further that our algorithm runs in time $O(n^2\log n + nk \log k).$ Let us analyze each step of the algorithm.
The first step of the algorithm is constructing the shortest path map SPM($x$), where $x$ is an arbitrary point of the boundary of $\mathcal K.$ In order to do this, we apply the algorithm presented in the previous section, which builds SPM($x$) in $O(n^2\log n)$ time using a data structure of size $O(n^2)$. \\
For each point $s$ of $S$, the implemented data structure contains the coordinates of this point inside the face $F$ of $\mathcal K$ that contains $s$ and the list $\mathcal L(F)$ of the intersection points of the boundary of $F$ with cones of SPM($x$). Therefore, using a binary search, we can determine the cone of SPM ($x$) containing $s$ in time $\log n$ (the size of $ \mathcal L(F)$ is $O(n)$). Since $S$ contains $k$ points, this step is executed in $O(k\log n)$ time. \\
At this stage, we want to compute the convex hull of each subset $S_i$ belonging to a cone of SPM($x$).
In order to do this, we perform an isometric embedding in the plane of the faces of cones containing points of $S_i$.
This embedding can be done in time $O(n)$ (Lemma \ref{liniar2}). \\
For each subset $S_i$, the algorithm builds in the plane the convex hull of the images of points of $S_i$ in $O(k \log k)$ time (for example using the incremental algorithm \cite{PrSh}).
Since SPM($x$) contains $O(n)$ cones (Lemma \ref{liniar1}), the number of subsets $S_i$ is $O(n).$ Therefore, the construction of convex hulls of all subsets of $S$ is performed in total time $O(nk \log k).$

\begin{figure}[h]\centering
\includegraphics[width=0.99\textwidth]{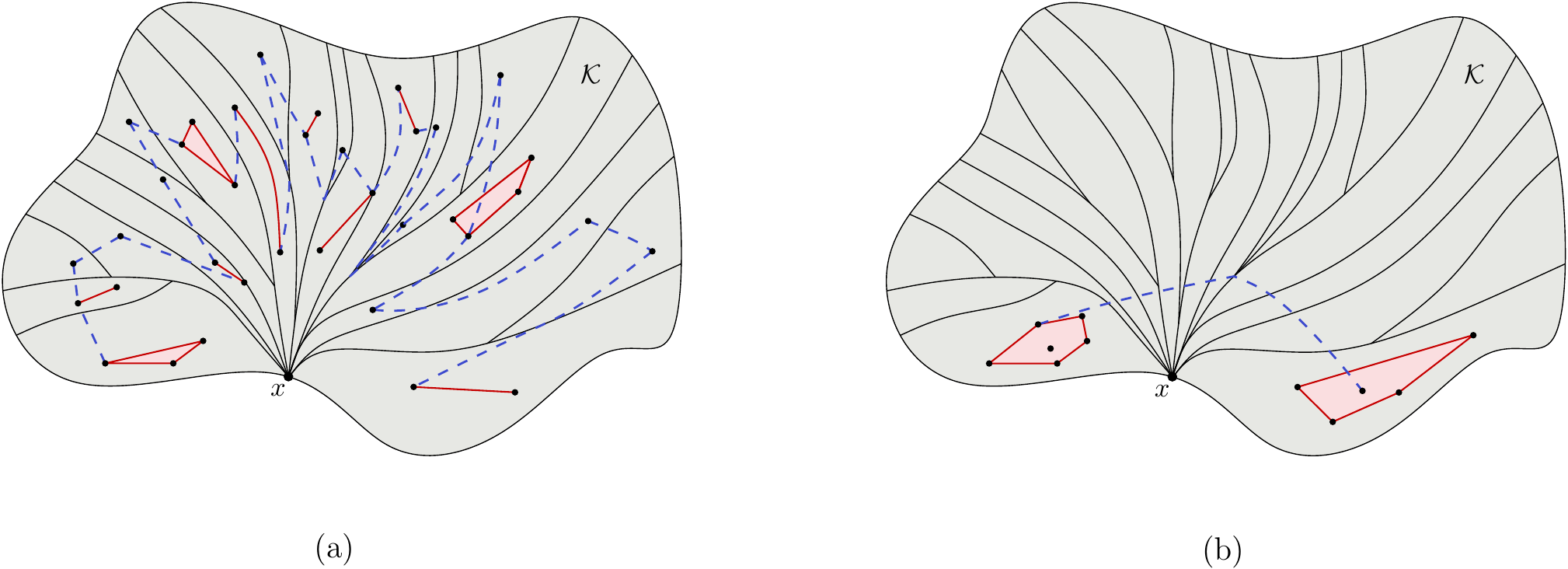}
\caption{The construction of geodesic segments between conv$(S_i)$ and conv$(S_{i+1})$.}\label{descomplexitate}
\end{figure}

The most expensive step of the algorithm is the construction of geodesic segments connecting the convex hulls conv$(S_i)$ and conv$(S_{i+1})$.
In order to construct a geodesic connecting two points in a CAT(0) planar complex, we use the algorithm presented previously. This algorithm computes the shortest path between two points of $\mathcal K$ in $O(n^2\log n)$ time using a data structure of size $O(n^2)$.
We want to calculate the total execution time of this step.
If two subsets $S_i$ and $S_{i+1}$ belong to two consecutive cones of SPM($x$), then to construct the geodesic segment between $s' \in S_i$ and $s'' \in S_{i+1} $ it is sufficient to perform an isometric embedding of these two cones in the plane, and construct the Euclidean segment between the images of $s'$ and $s''$. The embedding of two consecutive cones is possible, since the interior angle with the origin in the apex of a cone is at most $\pi$. \\
If every cone of SPM($x$) is containing at least one point of $S$ (see Fig. \ref{descomplexitate} (a)), and since SPM($x$) contains $O(n)$ cones, the construction of all geodesics between two points of two consecutive subsets is executed in $O(n^2\log n)$ time. \\
Suppose that the set $S$ is divided by SPM($x$) into two subsets $S_1$ and $S_2$, such that they belong to two extreme cones of SPM($x$) (see Fig. \ref {descomplexitate} (b)). In order to construct the geodesic between two points $s' \in S_1$ and $s'' \in S_2, $ we use the same algorithm presented previously which builds $\gamma(s', s'')$ in $O(n^2\log n)$ time.

\noindent We claim that in the case where $S$ is divided into $t<n$ subsets $S_i$, $i = 1, \ldots, t,$ the execution time of this step remains $O(n^2\log n).$
Let $y$ and $z$ be two points belonging respectively to the cone $\mathcal C_i$ and $ \mathcal C_j $ ($ 1 <i <j <n $) considered in their order starting from an extreme cone of SPM($x$). Since the boundary is a geodesic cone, the segment connecting $y$ and $z$ belongs to the subcomplex denoted $\mathcal K_{ij}$ having its extreme cone $\mathcal C_i$ and $\mathcal C_j$. Therefore, the construction of $\gamma(y, z)$ requires $O({n'}^2\log n')$ time with respect to the number of vertices $n'$ of the subcomplex $\mathcal K_{ij}$. \\
Computing all the geodesic segments connecting each pair of subsets $S_i$ and $S_{i+1}$ for all $i = 1, \ldots, t-1$ requires a total time of $O(n^2\log n)$.

\noindent Finally, using the same algorithm of the previous section, we compute the geodesic $\gamma(p, p^*)$ inside the complex $\mathcal K \setminus P$ in $O(n^2\log n)$ time.

In conclusion, we have shown that using a data structure of size $O(n^2 + k)$ for a set $S$ of $k$ points inside $\mathcal K $ with $ n $ vertices, our algorithm computes the convex hull conv$(S)$ in total time $O(n^2\log n + nk \log k)$. \hfill $\Box$
\end{proof}

\section{Conclusions}
In this paper we presented two algorithmic problems for the case of a CAT(0) planar complex $\mathcal K$.
The first problem is computing the shortest path between a given point $x$ and any query point $y$ in $\mathcal K$. To solve this problem we proposed in a method based on continuous Dijkstra algorithm \cite{HerSu, Mi1, ReSto}. \\
The preprocessing step of our algorithm is to construct SPM($x$) in $O(n^2\log n)$ time by sweeping the faces of the complex. For any query point $y$ the algorithm determines the cone of SPM($x$) containing $y$, then unfolds it in the plane and computes the shortest path $\gamma(x, y)$ between $ x $ and $ y $ in $ \mathcal K $. The shortest path $\gamma(x,y)$ is calculated in $O(n)$ time as the isometric image of the shortest path between the images of $ x $ and $ y $ in the plane. \\
The second problem presented in this paper for a complex CAT (0) planar is the construction of the convex hull of a finite set of points. \\
We proposed an algorithm that constructs the convex hull of a set $S$ of $k$ points in a CAT(0) planar complex with $n$ vertices in $O(n^2\log n + nk \log k)$ time using a data structure of size $O (n^2 + k). $ Our algorithm is similar to the algorithm of Toussaint \cite{To} for constructing the convex hull of a finite set of points in a triangulated simple polygon in $O(n \log n)$ time, using a data structure of size $ O(n) $.

An open question concerns the improvement of the running time of the shortest path algorithm.\\
(1) We do not know how to design a subquadratic algorithm which will compute the shortest path map in a CAT(0) planar complex.
Such an algorithm will improve considerably the running time of the two algorithms for the shortest path problem and for the convex hull problem which are presented in this article.

It would be interesting to study the structural properties of convex sets in the CAT(0) 2-dimensional complex and especially in non-planar CAT(0) complexes. Intuitively, the structure of a convex set in these complexes is rather complicated.
Open algorithmic question is to study the problem of construction of the convex hull of a finite set of points in general CAT(0) complexes.
\medskip\noindent
(2) It will be interesting to generalize our algorithmic results to general CAT(0) complexes, in particular 2-dimensional non-planar CAT(0) complexes.

\section*{Acknowledgement}

\noindent
I wish to thank Victor Chepoi for his useful advice and guidance,
and the anonymous referees for a careful reading of the first version of the manuscript and useful suggestions.
This work was supported in part by the ANR grants  OPTICOMB (ANR BLAN06-1-138894) and GGAA.


\begin{appendices}

\section{Proof of Proposition \ref{pr-spm}}\label{app}

\textbf{(i)} We associate to SPM($x$) an oriented graph $T(x) = (V(T), E(T))$ such that $V(T)$ contains the vertices of $\mathcal K$ together with the intersection points of the geodesics of SPM($x$) with $\partial\mathcal K$ (see Fig. \ref{des_spm-tree}). The set of edges $E(T)$ of $T(x)$ is such that two vertices $u, v$ of $T(x)$ are connected by an oriented edge $\overrightarrow{e} = \overrightarrow{uv}$ if one of the geodesics of SPM($x$) passes via $u$ and $v$, $d(x, u)<d(x, v)$ and the geodesic segment $\gamma(u, v)$ contains no other vertices of $\mathcal K.$

We now show that $T(x)$ is a tree. In order to prove this, we must show that for every vertex $v$ of $T(x)$ there exists a unique edge $e$ incident to $v$ and oriented towards $v$.
\begin{figure}[h]\centering
\includegraphics[width=0.9\textwidth]{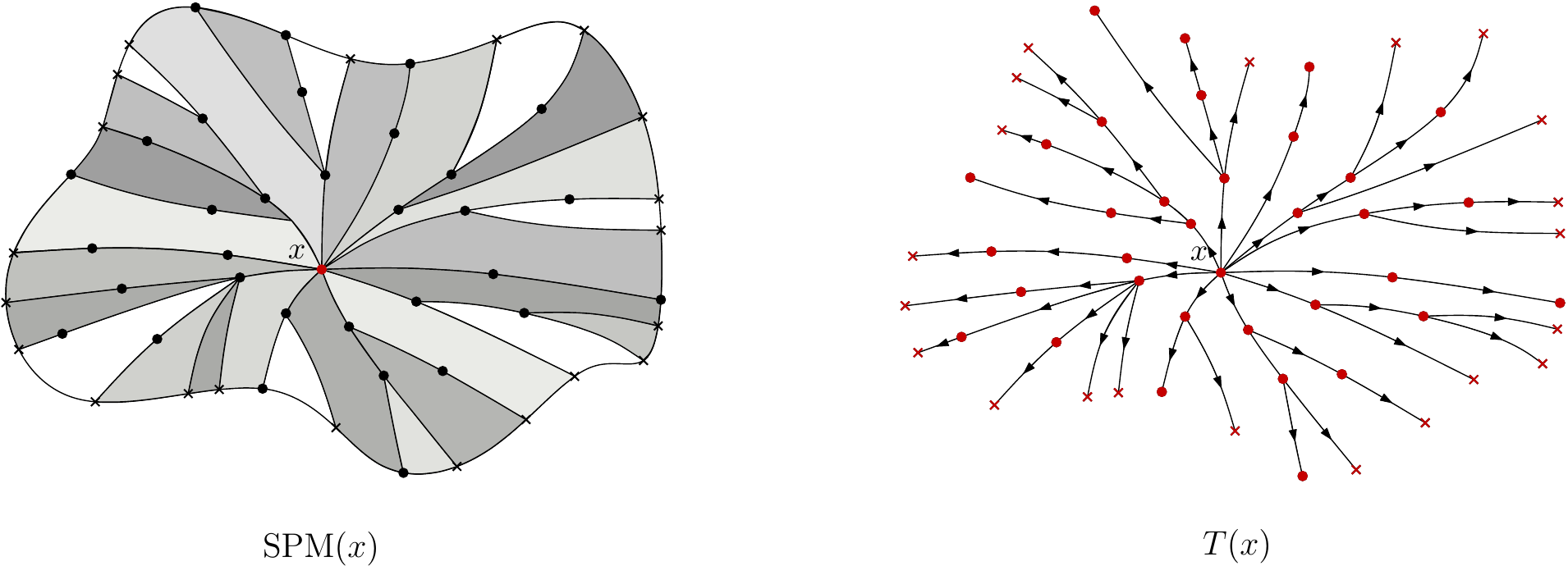}
\caption{The oriented graph $T(x)$ associated to SPM($x$).}\label{des_spm-tree}
\end{figure}

Suppose by contradiction that for a vertex $u$ of $T(x)$ there are at least two edges $\overrightarrow{e_1}$ and $\overrightarrow{e_2}$ incident to $u$ and oriented towards $u$.
Every edge of $T(x)$ represents a geodesic segment belonging to a geodesic of SPM($x$) passing via $x$. Thus the edges $\overrightarrow{e_1}$ and $\overrightarrow{e_2}$ are two geodesics $\gamma_1$ and $\gamma_2$ of SPM($x$) containing the origin $x$. On the other hand, because $\overrightarrow{e_1}$ and $\overrightarrow{e_2}$ are incident to $u$, the geodesics $\gamma_1$ and $\gamma_2$ passe via $u$ which is a contradiction with the uniqueness of the geodesic between any two vertices $x$ and $u$ in $\mathcal K.$

\textbf{(ii)} Suppose for the sake of contradiction that for a cone $\mathcal C (z; p, q) $ of SPM($x$), the inner angle $\angle_z(p, q)$ with the origin in the apex of the cone is greater than or equal to $\pi$.\\
First we want to show that the apex $z$ of $\mathcal C(z; p, q)$ is necessarily a vertex of $\mathcal K$. Suppose that $z$ is not a vertex of the complex.
By property (2) of SPM, $z$ belongs to at least two geodesic segments $\gamma(x, p)$ and $\gamma(x, q)$ of SPM($x$), where $p,q\in\partial\mathcal K$ (see Fig. \ref{des_unghi-con}).
The points $p$ and $q$ are distinct, thus we consider without loss of generality that $z$ is the furthest point from $x$ such that $z\in \gamma(x, p)\cap\gamma (x, q).$ \\
Thus $z$ can be either an inner point of a face of $\mathcal K$ or $z$ belongs to an edge of $\mathcal K $. In both cases $z$ is a point of zero curvature.
Suppose that $z$ is a point of a face $\Delta(a, b, c)$ of $\mathcal K$. The face $\Delta(a, b, c)$ is an Euclidean triangle and since $z$ is the last common point $\gamma(x, p)$ and $\gamma(x, q)$, we obtain that the angles $\angle_z(x, p)$ and $\angle_z(x, q)$ are each less than $\pi$, and so $\gamma (x, p)$ and $\gamma(x, q )$ are not locally shortest paths. Therefore, $\gamma(x, p)$ and $\gamma(x, q)$ are not geodesics (contrary to the hypothesis).
Thus the apex of a cone of SPM($x$) is necessarily a vertex of $\mathcal K.$

Since $z$ is a vertex of $\mathcal K$, there exists a face $F = \Delta(z, a, b)$, where $a, b$ are two vertices of $\mathcal K$ such that $(F\setminus\{z\})\cap \mathcal C(z; p, q)\neq \emptyset$. By the definition of SPM($x$) (1), we can deduce that the vertices $a$ and $b$ do not belong to the interior of $\mathcal C(z;p,q)$, otherwise $\mathcal C(z;p,q)$ is replaced by smaller cones. Thus, the angle $\angle_z(a,b)$ can only be greater than or equal to the angle $\angle_z(p,q)$.
Since $F$ is isometric to a triangle of the plane $\Delta'= \Delta(z', a', b')$ and by the Alexandrov's property of angles in a CAT(0) metric space (mentioned in section \ref{prel1}), we obtain that $\angle_z(a, b)\leq \angle_{z'}(a', b')\leq \pi$. Which contradicts our assumption. \\

\textbf{(iii)} Suppose to the contrary that there exist two points $r, s$  of the cone $\mathcal C(z; p, q)$ of SPM($x$) such that the geodesic $\gamma(r,s)$ contains at least one point $y$ outside the cone. Therefore, the geodesic $\gamma(r, s)$ intersects the boundary of the cone at least twice. Let $a$ and $b$ be two consecutive intersection points of the geodesic $\gamma(r, s)$ with boundary $\mathcal C(z; p, q)$ such that $y$ belongs to the geodesic segment $\gamma(a,b).$
By our assumption, $\gamma(r, s)$ contains the points $a, y$ and $b$ in this order, where $a, b \in \partial \mathcal C(z; p, q)$ and $y \not \in \mathcal C(z; p, q).$ \\
We consider the following two cases: (a) where $a$ and $b$ belong to two different sides of the cone (e.g. $a \in \gamma (z,p)$ and $b \in \gamma(z, q)$) and (b) the points $a$ and $b$ belong to one side of the cone $ \mathcal C(z; p, q)$ (e.g. $ a, b \in \gamma(z, p)$).
\medskip
\begin{figure}[h]\centering
\begin{tabular}{cc}
   \includegraphics[width=0.45\textwidth]{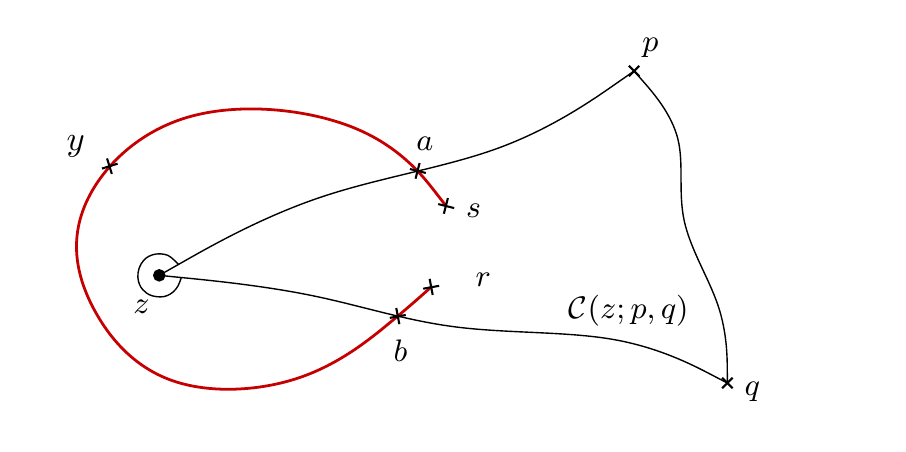}
   &
   \includegraphics[width=0.45\textwidth]{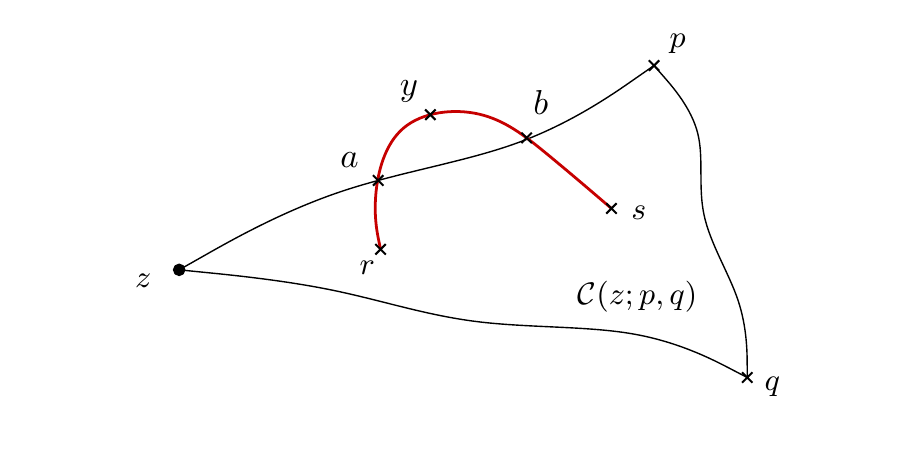}
   \\
   (a)
   &
   (b)
\end{tabular}
\caption{Illustration of the case: $\gamma(r,s)\not\subset\mathcal C(z;p,q)$.} \label{des_con-convex}
\end{figure}
\medskip

First, let us analyze the first case (see Fig. \ref{des_con-convex} (a)). Let $a \in \gamma(z, p)$ and $b \in \gamma(z, q),$ then we obtain a geodesic triangle $\Delta(a, z, b)$ where $y$ is a point of the edge of the triangle defined by $a$ and $b$.
By the property (ii), the angle $\angle_z(p, q)$ in the cone $\mathcal C(z; p, q)$ is less than $\pi.$ Since for any vertex $z$ of $\mathcal K$, $\theta(z) \geq 2\pi $, then the complementary angle $\angle_z(p, q)$ which does not belong to $\mathcal C (z; p, q),$ is greater than or equal to $\pi$.
By the Alexandrov property of angles the sum of the angles of a geodesic triangle in a CAT(0) metric space is at most $\pi$ \cite{BrHa}. On the other hand, the sum of the angles of $\Delta(a, z, b)$ containing the angle $\angle_z(p, q)$ external to the cone $\mathcal C(z; p, q)$ is strictly greater than $\pi$. This contradiction shows that this case is impossible.

Let us analyze the second case. Let $a, b \in \gamma (z, p) \cap \gamma(r, s)$ and the point $y$ such that $y \in \gamma(r, s)$ and $y \not \in \mathcal C(z; p, q).$ Therefore, the points $a$ and $b$ are connected by two distinct geodesics (see Fig. \ref{des_con-convex} (b)): the geodesic $ \gamma(a, b) \subset \gamma(z, p)$ and the geodesic $\gamma (a, b) \subset \gamma(r, s)$ passing via the point $y.$ We obtain a contradiction of the uniqueness of a geodesic connecting two points in a CAT(0) metric space.

\textbf{(iv)} Suppose the contrary and let $u$ be a vertex of $\mathcal K$ with $u$ an inner point of the cone $\mathcal C(z; p, q)$ of SPM($x$).
According to the definition of SPM($x$), there exists a geodesic $\gamma(x, r)$ of $\mathcal K$ belonging to SPM($x$), which passes via $u$, where $r \in \partial \mathcal K$. Since any geodesic of SPM($x$) is a side of at least one cone of SPM($x$), this contradicts the existence of the cone $\mathcal C(z; p, q)$ in SPM($x$). \\

\begin{figure}[h]\centering
\includegraphics[width=0.45\textwidth]{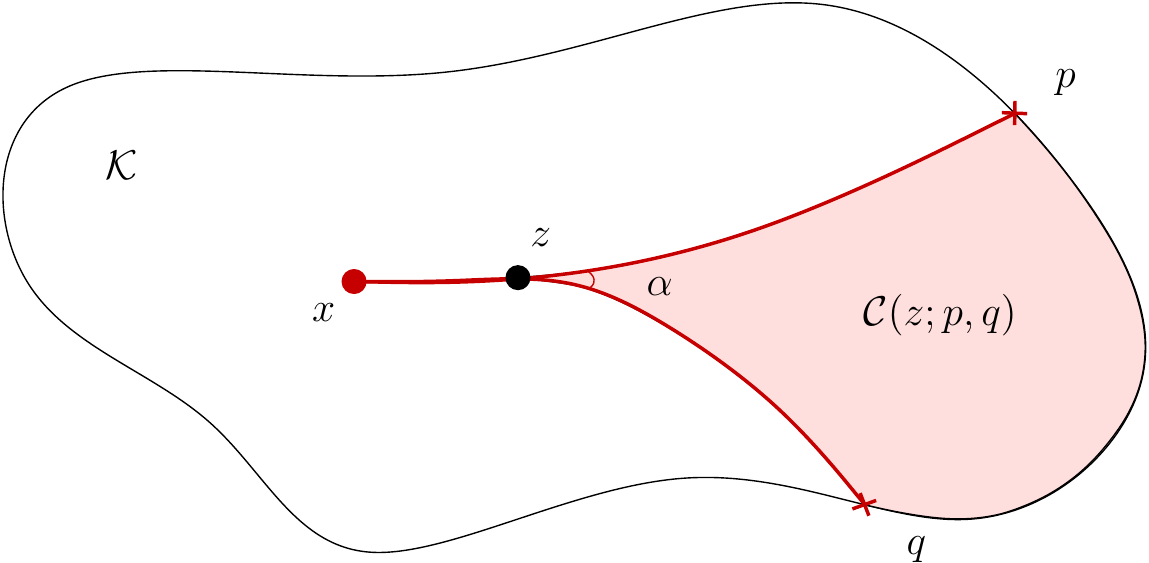}
\caption{The angle of a cone of SPM($x$).}\label{des_unghi-con}
\end{figure}

\textbf{(v)} According to the definition of SPM($x$), the complex $\mathcal K$ is partitioned into convex cones. Thus, for any point $y$ of $\mathcal K,$ $y$ belongs to a cone $\mathcal C(z; p, q)$ of the SPM($x$) where one of the scenarios are possible: $y\in $int$(\mathcal C (z; p, q)),$ in which case $y$ belongs to a single cone, or $y \in\partial C(z; p, q),$ in which case $y$ belongs to at least two cones (to exactly two cones if $y$ is of nonnegative curvature and to three or more cones if $y$ is of negative curvature).\\

\textbf{(vi)} By the definition of SPM($x$), the sides of the cone $\mathcal C(z; p, q)$ are the geodesics $\gamma(z, p)$ and $\gamma(z, q)$.
Let $u\in \gamma(z, p).$ From the definition of a geodesic, in a neighborhood $B(u, \epsilon) = \{y \in \mathcal K: d(u, y) \leq \epsilon \}$ of $u$ $\gamma(z, p)$ is a geodesic. Let $a, b$ be two points of $\mathcal K$ for which $C(u, \epsilon) \cap \gamma(z, p) = \{a, b \},$ where $C(u,\epsilon) = \{y \in \mathcal K: d(u, y) = \epsilon \}$. Consequently, the inner angle $\angle_u(a, b)$ in $\mathcal C(z; p, q)$ is greater or equal to $\pi.$ Since $a, b \in \gamma(z, p)$ then $\angle_u(z, p) \geq \pi$. 
\end{appendices}
\end{document}